\documentclass[10pt, reqno]{amsart}
\usepackage{amsmath, amssymb}
\usepackage{amsfonts}
\usepackage{mathrsfs}

\usepackage{amsthm}

\usepackage{float}
\usepackage{hyperref}
\usepackage{graphics}
\usepackage{graphicx}
\usepackage{verbatim}
\usepackage{multirow}
\usepackage{tikz}
\usetikzlibrary{arrows,decorations.pathmorphing,matrix,chains,positioning}
\usepackage{enumerate}
\usepackage{eufrak}
\usepackage{bbm}
\usepackage{anysize}
\marginsize{42.5mm}{42.5mm}{48.5mm}{35mm}
\allowdisplaybreaks

\newtheorem{theorem}{Theorem}[section]
\newtheorem{lemma}[theorem]{Lemma}
\newtheorem{proposition}[theorem]{Proposition}
\newtheorem{corollary}[theorem]{Corollary}

\newtheorem*{theorem*}{Theorem}
\theoremstyle{remark}

\newtheorem{definition}[theorem]{Definition}

\newenvironment{dispmathenv}%
{\stepcounter{equation}

 \medskip

}{%
\medskip%
\hfill(\theequation)\newline
}

\numberwithin{equation}{section}

\newcommand{\N}{\mathbb{N}}
\newcommand{\R}{\mathbb{R}}

\newcommand{\s}{\mathbb {S}}
\newcommand{\E}{\mathbb{E}}
\newcommand{\prob}{\mathbb{P}}
\newcommand{\W}{\mathcal {W}}
\newcommand{\K}{\mathcal {K}}
\newcommand{\p}{\mathcal {P}}
\newcommand{\Ordo}{\mathcal {O}}

\begin{document}
\title[Propagation of Chaos for the Thermostatted Kac Master
equation]{Propagation of Chaos for the Thermostatted Kac Master
  equation }
\author{Eric Carlen, Dawan Mustafa, Bernt Wennberg}
%
%
%
%
%
%
%
\begin{abstract}
 The Kac model is a simplified model of an $N$-particle system in which the collisions of a real particle system are modeled by
 random jumps of pairs of particle velocities. Kac proved propagation of chaos for this model, and hence provided a rigorous
 validation of the corresponding Boltzmann equation. Starting with the same model we consider an $N$-particle system in which the particles
 are accelerated between the jumps by a constant uniform force field which conserves the total energy of the system.
 We show propagation of chaos for this model.
\end{abstract}
\maketitle


\section{Introduction}
\noindent
The most fundamental equation in the kinetic theory of gases is perhaps the Boltzmann equation, which was derived
by Ludwig Boltzmann in 1872. This equation describes the time evolution of the density of a single particle in a gas
consisting of a large number of particles and reads
\begin{equation}\label{boltzeq}
\frac{\partial}{\partial t} f(x,v,t)+v \cdot \nabla_x f(x,v,t)=Q(f,f)\,
\end{equation}
where $f(x,v,t)$ is a density function of a single particle, $x, v\in \R^3$ represent the position and
velocity of the particle, and $t \geq 0$ represents time.
The collision operator $Q$ is given by
\begin{equation}\label{colboltzmann}
Q(f,f)=\int_{\R^3}\int_{\s^2}[f(x,v',t)f(x,v_*',t)-f(x,v,t)f(x,v_*,t)]
B(v-v_*,\sigma)\mathrm{d}v_*\mathrm{d}\sigma.
\end{equation}
The case where $f$ is independent of $x$ is called the spatially homogeneous Boltzmann equation.
In equation (\ref{colboltzmann}), the pair $(v,v_*)$ represents the
velocities of two particles before a collision and $(v',v_*')$ the
velocities of these particles after the collision.

The fact that the collision operator $Q(f,f)$ involves products of the density
$f$ rather than a two particle density $f_2(x_1,v_1,x_2,v_2,t)$ is a consequence of Boltzmann's~ \emph{stosszahlansatz},
the assumption that two particles engaging in a collision are independent before the interaction.
It is a very challenging problem to improve on Landford's result from $1975$~\cite{lanford}, which essentially states
that the stosszahlansats holds for a time interval of the order of one fifth of the time mean time between collisions of an individual particle.
In an attempt to address the fundamental questions concerning the derivation of spatially homogeneous the Boltzmann equation,
Mark Kac introduced  a stochastic particle process consisting of $N$
particles from which he obtained an equation like the Boltzmann
equation (\ref{boltzeq}) as a mean field limit when the numbers of the
particles $N\rightarrow \infty$, (see~\cite{kac}): Consider the
\emph{master vector}  $\textbf{V}=(v_1,\dots,v_N)$, $v_i\in \R$, where each
coordinate represents the velocity of a particle. The spatial
distribution of the particles is ignored in this model, and the
velocities are one dimensional.
The state space of the particles is
the sphere in $\R^N$ with radius $\sqrt{N}$, that is the velocities
are restricted to satisfy the equation
\begin{equation}
v_1^2+\dots+v_N^2= N.
\end{equation}
The binary collisions in a the  gas are  represented by jumps
involving pairs of velocities from the master vector, with
exponentially distributed time intervals with intensity $1/N$. At each
collision time, the pair of velocities $(v_i, v_j)$ are chosen
randomly from the master vector and changes to $(v_i',v_j')$ according
to
\begin{align*}
v_i'=v_i\cos\theta+v_j\sin\theta,\\
v_j'=v_j\cos\theta-v_i\sin\theta.
\end{align*}
The parameter $\theta$ is chosen according to a law
$b(\theta)\mathrm{d}\theta$. In~\cite{kac}, for the sake of
simplicity, Kac chooses $b(\theta)=(2\pi)^{-1}$. Any bounded
$b(\theta)$ can be treated in the same way. The post-collision master vector is denoted by
$R_{ij}(\theta)\textbf{V}$. Note that  the collision process does not
conserve both momentum and
energy (only trivial collisions can conserve both invariants in this one
dimensional case). Hence
$$v_i^2 + v_j^ 2=v_i'^2 + v_i'^2,$$
but in general
$$v_i+ v_j\ne v_i' + v_i'.$$
The equation governing the evolution of this process is called Kac's master equation
(a Kolmogorov forward equation for Markov processes). It is given  by
\begin{equation}\label{kacmastereq}
      \frac{\partial}{\partial t}W_N(\textbf{V},t)=\K
      W_N(\textbf{V},t).   \qquad
      W_N(\textbf{V},0) = W_{N,0}(\textbf{V})
  \end{equation}
where the collision operator $\K$ has the form
\begin{equation}\label{coloperator}
\K W_N(\textbf{V},t)=\frac{2}{N-1}\sum_{1\leq i<j\leq N}(Q_{(i,j)}-I)W_N(\textbf{V},t),
\end{equation}
with
\begin{equation}\label{Qij}
Q_{(i,j)}W_N(\textbf{V},t)=\int_{-\pi}^{\pi}W_N((R_{ij}(\theta)\textbf{V}),t)\frac{\mathrm{d}\theta}{2\pi}.
\end{equation}
The particles are assumed to be identical and this corresponds to the initial density being symmetric:
\begin{definition}
A probability density $W(\textbf{V})$ on $\R^N$ is said to be
symmetric if for any bounded continuous function $\phi$ on $\R^N$
\begin{equation}
\int_{\R^N}\phi(\textbf{V})
W(\textbf{V})\mathrm{d}m^{(N)}=\int_{\R^N}\phi(\textbf{V}_{\sigma})
W(\textbf{V})\mathrm{d}m^{(N)}
\end{equation}
where for any permutation $\sigma \in \{1,\dots,N\}$
$$V_\sigma=(v_{\sigma(1)},\dots,v_{\sigma(N)}).$$
\end{definition}
\noindent We note that the master equation (\ref{kacmastereq})
preserves symmetry. To obtain an equation like (\ref{boltzeq}) which
describes the time  evolution
of a one-particle density, Kac studied the $k$-th marginal $f_k^N$ of $W_N(\textbf{V},t)$, where
\begin{equation}
f_k^N(v_1,\dots,v_k,t)=\int_{\Omega_k}W_N(\textbf{V},t)\mathrm{d}\sigma^{(k)}.
\end{equation}
Here, $\sigma^{(k)}$ is the spherical measure on
$\Omega_k=\s^{N-1-k}\left(\sqrt{N-(v_1^2+\dots+v_k^2)}\right)$.
Since $W_N$ is symmetric, the $k$-th marginal is also symmetric, and the
time evolution for the first marginal $f_1^N$ is obtained by
integrating the master equation (\ref{kacmastereq}) over the variables
$v_2 \dots v_N$. This yields
\begin{equation}\label{evfirstmargkac}
\frac{\partial}{\partial
  t}f_1^N(v_1,t)=2\int_{-\sqrt{N-v_1^2}}^{\sqrt{N-v_1^2}}\int_{-\pi}^{\pi}
[f_2^N(v_1',v_2',t)-f_2^N(v_1,v_2,t)]\frac{\mathrm{d}\theta}{2\pi}\mathrm{d}v_2,
\end{equation}
where
$$ v_1'=v_1\cos\theta + v_2\sin\theta,\hspace{0.2in} v_2'=v_2\cos\theta - v_1\sin\theta.$$
If we had $f_2^N(v_1,v_2,t)\approx f_1^N(v_1,t)f_1^N(v_2,t)$ in a weak
sense (which is defined later) then the evolution equation (\ref{evfirstmargkac}) for the first marginal
would look like the spatially homogenous Boltzmann equation, i.e., equation (\ref{boltzeq}) without the position variable $x$. Kac
suggested in~\cite{kac} that one should take a sequence of initial
densities $W_{N,0}(\textbf{V})$ which have the ``Boltzmann property'' that is,
$$\lim_{N\rightarrow
  \infty}f_k^N(v_1,\dots,v_k,0)=\prod_{j=1}^{k}\lim_{N\rightarrow
  \infty}f_1^N(v_j,0),$$
weakly in the sense of measures on $\R^k$. The Boltzmann property means that for each fixed $k$, the joint
probability densities of the first $k$ coordinates tend to product
densities when
$N\rightarrow \infty$.
By analyzing how the collision operator acts on functions depending
on finitely many variables  and a combinatorial argument Kac showed
that for all $t>0$, the sequence $W_N(\textbf{V},t)$ also has the
Boltzmann property, that is, the Boltzmann property propagates in
time. In this case the limit of the first marginal
$f(v,t)=\lim_{N\rightarrow \infty}f_1^N(v,t)$ satisfies the
Boltzmann-Kac equation
\begin{equation*}
\frac{\partial}{\partial t}f(v,t)=\mathcal{Q}(f,f),
\end{equation*}
where
\begin{equation}\label{colkacboltzmann}
\mathcal{Q}(f,f)(v)=2 \int_\R\int_{-\pi}^{\pi}[f(v',t)f(u',t)-f(v,t)f(u,t)]
  \frac{\mathrm{d}\theta}{2\pi}\mathrm{d}u.
\end{equation}
What Kac refereed to as the 'Boltzmann property' is nowadays often called
\emph{chaos}. More precisely, we have the following definition:
\begin{definition}
Let $f$ be a given probability density on $\R$ with respect to the
Lebesgue measure $m$. For each $N\in \N$, let $W_N$
be a probability density on $\R^N$ with respect to the product measure
$m^{(N)}$. Then the sequence $\{W_N\}_{N\in \N}$ of probability
densities on $\R^N$ is said to be $f$-{\bf chaotic} if
\begin{enumerate}
  \item Each $W_N$ is a symmetric function of the variables $v_1,v_2,\cdots,v_N$.
  \item For each fixed $k\in \N$ the k-th marginal
    $f_k^N(v_1,\dots,v_k)$ of $W_N$ converges to
    $\prod_{i=1}^{k}f(v_i) $, as $N\rightarrow \infty$,
   ($f(v)=\lim_{N\rightarrow \infty}f_1^N(v)$) in the sense of weak
   convergence, that is, if $\phi(v_1,v_2,\dots,v_k)$ is bounded
   continuous function on $\R^k$, then
      $$\lim_{N\rightarrow \infty}\int_{\R^N}\phi(v_1,v_2,\dots,v_k)W_N(\textbf{V})\mathrm{d}m^{(N)}
       =\int_{\R^k} \phi(v_1,v_2,\dots,v_k) \prod_{i=1}^{k}f(v_i)\mathrm{d}m^{(k)}.$$
\end{enumerate}
\end{definition}
The aim of this paper is to show propagation of chaos for a new many particle model with the
same collision process, but where between the collisions, the particles are accelerated
by a force field which always keep the total energy constant. In the
next subsection we describe this process.  In the original problem
considered by  Kac~\cite{kac}, correlations between particles were
only introduced through the binary collisions. In our case, the force
field will introduce correlations as well, but of a different character.\\
\indent Our proof of propagation of chaos for this model with two
distinct sources of correlation  builds on recent work on propagation
of chaos, but also includes a {\em quantitative}
development
of Kac's original argument  which we apply to control the correlations
introduced by the collisions. We must quantify these correlations in
order to
control the correlating effects of the force field.
\subsection{The Thermostatted Kac master equation.\\ \\}
\noindent In the Kac model, the particles interact via random jumps
which correspond to random collisions between pairs of particles. We
now consider
a stochastic model where the particles have the same jump process as
in the Kac model, but are now also accelerated between the jumps under
a constant uniform force field $\textbf{E}=E(1,1,\dots,1)$ which
interacts with a Gaussian thermostat in order to keep the total energy
of the system constant.
For a detailed discussion see~\cite{yosief}. Consider the master
vector $\textbf{V}=(v_1,\dots,v_N)$ on the sphere $\s^{N-1}(\sqrt{N})$. The vector
$\textbf{V}$ clearly depends on time , and when needed  we
write $\textbf{V}(t)$ instead of $\textbf{V}$. It is also convenient
to use a coordinate-system in which $E>0$. The Gaussian thermostat is
implemented as the projection of $\textbf{E}$ into the tangent plane
of $\s^{N-1}(\sqrt{N})$ at the point $\textbf{V}$. The time evolution
of the master vector between collisions is then given by :
\begin{equation}\label{Nparticled}
\frac{d}{dt} \textbf{V} = \textbf{F}(\textbf{V}),
\end{equation}
where
\begin{equation}\label{Nparticle_Ffield}
\textbf{F}(\textbf{V})=E \left (\textbf{1}-\frac{J(\textbf{V})}{U(\textbf{V})}\textbf{V}\right)
\end{equation}
and $\textbf{1}=(1,\dots,1)$. The quantities $J(\textbf{V})$ and
$U(\textbf{V})$ represent the \emph{average momentum} per particle and
the \emph{average energy} per particle, respectively, and are given by
\begin{equation}
J(\textbf{V})=\frac{1}{N}\sum_{i=1}^{N}v_i,
\end{equation}
\begin{equation}
U(\textbf{V})=\frac{1}{N}\sum_{i=1}^{N}v_i^2.
\end{equation}
If  $W_N(\textbf{V},t)=W_N(v_1,v_2,\dots,v_N,t)$ is the probability
density of the $N$ particles at time $t$, it satisfies the so called
Thermostatted Kac master equation (see~\cite{yosief})
\begin{equation}\label{Nparticlemaster}
\frac{\partial}{\partial
  t}W_N(\textbf{V},t)+\nabla_{\textbf{V}}\cdot(\textbf{F}(\textbf{V})
W_N(\textbf{V},t))=\K(W_N)(\textbf{V},t).
\end{equation}
We see that (\ref{Nparticlemaster}), in the absence of the force field
reduces to the master equation of the Kac model. Under the
assumption that the sequence of probability densities
$\{W_N(\textbf{V},t) \}_{N\in \N}$ propagates chaos it is shown
in~\cite[Theorem 2.1]{yosief} that
$f(v,t)=\lim_{N\rightarrow \infty}f_1^N(v,t)$ where $f_1^N(v,t)$ is
the first marginal of $W_N(\textbf{V},t)$ satisfies the Thermostatted
Kac equation
\begin{equation}\label{thermostat_boltzmann}
\frac{\partial}{\partial t}f(v,t)+E\frac{\partial}{\partial
  v}\left((1-\zeta(t)v)f(v,t)\right)=\mathcal{Q}(f,f),
\end{equation}
where
\begin{equation}
\zeta(t)=\int_\R vf(v,t)\mathrm{d}v,
\end{equation}
and $\mathcal{Q}(f,f)$ is given by (\ref{colkacboltzmann}).
For the investigation of equation (\ref{thermostat_boltzmann}) we refer to
Wennberg, Wondmagegne~\cite{berntyosef} and Bagland~\cite{vb}.
\\ \\
The interest in studying thermostatted kinetic equations comes from attempts to fully understand Ohm's law.
Many of the ideas of this paper come from~\cite{bcelm}, which presents a more
realistic model where  the positions of the particles are also taken
into account. However, the collision term is easier,
and  the main difficulty comes from analyzing a spatially homogenous model.\\
\indent The proof of propagation of chaos is in many
ways similar to that of Kac, but whereas
his proof is carried out entirely by analyzing the collision operator
(\ref{coloperator}), the proof presented here (and in~\cite{bcelm})
requires a more detailed analysis of the underlying stochastic jump
process, and thus approaches
Gr\"{u}nbaum's method for proving propagation of
chaos (see~\cite{grunbaum}), which is based on studying the empirical
measure $\mu_N$ generated by the $N$ velocities, and proving that the
sequence $\{\mu_N\}_{N=1}^\infty$ converges weakly to a measure, which
is the solution to the Boltzmann equation. While essentially all
ingredients of the proof are present in~\cite{grunbaum}, there are
many technical difficulties that were treated rigorously only later
in~\cite{cm}, and in a much greater generality in~\cite{berntmm} and
other papers by the same authors. A standard reference addressing many
aspects of the propagation of chaos is~\cite{berntmm}.
\\ \\
\indent The structure of the paper is as follows:
In Section $2$ we introduce a master equation (a ``quenched equation")
which is an approximation to the master equation
(\ref{Nparticlemaster}). In Section $3$, we show that the quenched
master equation propagates chaos with a quantitative rate. In Section
$4$ we make pathwise comparison of the stochastic processes
corresponding to the master equation (\ref{Nparticlemaster}) and the
approximation master equation. The main result is that, for large $N$,
the paths of the two stochastic processes are close to each other.
Finally, in Section $5$ we show that the second marginal of
$W_N(\textbf{V},t)$ converges as $N\rightarrow \infty$
to a product of two one marginals of $W_N(\textbf{V},t)$.
\section{An approximation process}
\noindent To show propagation of chaos for the evolution described by
the master equation (\ref{Nparticlemaster}), we consider the two
particle marginal $f_2(v_1,v_2,t)$ of $W_N(\textbf{V},t)$ and show
that it can be written as a product of $2$ one particle marginals of
$W_N(\textbf{V},t)$ when $N\rightarrow \infty$. In~\cite{bcelm}, by
introducing an approximation master equation which propagates
independence, it is shown that for large $N$, the path described by this
approximate master equation is close to the path described by the
original master equation. This in turn implies propagation of chaos.
The independence property is not crucial, and the ideas in~\cite{bcelm}
can be adapted and further developed so as to apply to the model we consider here. If one tries to directly
show propagation of chaos for the master equation
(\ref{Nparticlemaster}) using the classical method by Kac~\cite{kac},
one encounters difficulties, even with the master
equation in~\cite{bcelm}. The difficulty lies in the nature of the
force field $\textbf{F}(\textbf{V})$ which depends on $J(\textbf{V})$ and
$U(\textbf{V})$.\\
\indent To overcome this difficulty in~\cite{bcelm} a modified force
field is introduced in which the random quantities $J(\textbf{V})$ and
$U(\textbf{V})$ are replaced by their expectations which only depend
on time. This gives rise to a new master equation. In the next section we
introduce this modified problem and related properties.

\subsection{ The modified force field and the quenched master equation.\\ \\}
\noindent Following the lines in~\cite{bcelm}, given a probability
density on $\R^N$ we define the  \emph{quenched current} and the
\emph{quenched energy} approximation as:
\begin{equation}\label{defJhatUhat}
\widehat{J}_{W_N}(t) =\frac{1}{N}\sum_{j=1}^{N}<v_j>_{W_N(\textbf{V},t)} \hspace{0.2in}\mbox{and}\hspace{0.2 in}
\widehat{U}_{W_N}(t)
=\frac{1}{N}\sum_{j=1}^{N}<v_j^2>_{W_N(\textbf{V},t)},
\end{equation}
where $<\cdot>_{W_N}$ denotes the expectation with respect to a given
density $W_N$, i.e., for an arbitrary continuous function $\phi$,
\begin{equation*}
<\phi(\textbf{V})>_{W_N}=\int_{\R^N}\phi(\textbf{V}) W_N(\textbf{V})\mathrm{d} m^{(N)}
\end{equation*}
with $\mathrm{d}m^{(N)}$ denoting the Lebesgue measure on $\R^N$. The
modified force field which now depends on the quenched current and
energy is defined as
\begin{equation}\label{modffield}
\widehat{\textbf{F}}_{W_N}(t)=E \left
  (\textbf{1}-\frac{\widehat{J}_{W_N}(t)}{\widehat{U}_{W_N}(t)}{\textbf{V}}(t)
\right ).
\end{equation}
We note that with given $\widehat{J}_{W_N}(t)$ and
$\widehat{U}_{W_N}(t)$, the particles move independently when subject
to (\ref{modffield}), while in (\ref{Nparticle_Ffield}) all particles
interact through the force field $\textbf{F}$. With this modified
force field, we consider the following \emph {quenched master
  equation}
\begin{equation}\label{quenchmaster}
\frac{\partial}{\partial t}\widehat
W_N(\textbf{V},t)+\nabla\cdot(\widehat{\textbf{F}}_{\widehat
  W_N}(t)\widehat W_N(\textbf{V},t))
=\K \widehat W_N(\textbf{V},t),
\end{equation}
where now the modified force is the one corresponding to the density
$\widehat W_N(\textbf{V},t)$. Besides the difference in force
fields, the quenched master equation (\ref{quenchmaster}) is
non-linear ($\widehat{\textbf{F}}(t)$ depends on $\widehat{W}(t)$)
compared to the master equation (\ref{Nparticlemaster}) but they both
have the same collision process.  \\

The motivation for introducing the quenched process is that {\em if} there is propagation of chaos,
then the different particle velocities will be approximately independent, and then for large $N$, the Law of Large Numbers
will imply that almost surely,
\[
J(\textbf{V})=\frac{1}{N}\sum_{i=1}^{N}v_i(t)   \approx
 \widehat{J}_{\widehat W_N}(t) \ ,
\]
and likewise for the energy, to a very good approximation.  In this
case, there will be a negligible difference
between the quenched force field and the thermostatting force field,
and thus we might expect the two processes to be pathwise close.
To follow the strategy of~\cite{bcelm}, we shall need quantitative estimates on the
propagation of chaos by the quenched process, which shall justify using the Law of Large Numbers
to show that for large $N$ the two force fields are indeed close.\\

\noindent Henceforth, to simplify notations, let
$$\widehat J_N(t):=\widehat{J}_{\widehat W_N}(t)
\hspace{0.3in}\text{and}\hspace{0.3in} \widehat
U_N(t):=\widehat{U}_{\widehat W_N}(t).$$
In the next lemma we describe the time evolution of $\widehat J_N(t)$
and $\widehat U_N(t)$ in terms of differential equations.
\begin{lemma}\label{evJ}
Given initial distribution $\widehat W_{N,0}(\textbf{V})$,
$\widehat{J}_N(t)$ and
$\widehat{U}_N(t)$ satisfy the differential equations, both independent of $N$:
\begin{equation}\label{eqevJ}
\frac{\mathrm{d}}{\mathrm{d}t}\widehat{J}_N(t)=E-E
\frac{\widehat{J}_N(t)^2}{\widehat{U}_N(t)}-2\widehat{J}_N(t),
\end{equation}
and
\begin{equation}\label{eqevU}
\frac{\mathrm{d}}{\mathrm{d}t}\widehat{U}_N(t)=0.
\end{equation}
\end{lemma}
\begin{proof}
Formally the result can be obtained by multiplying
  eq.~(\ref{quenchmaster}) by $v_i$ or $v_i^2$, integrating
  (partially) and   summing over $i$. To avoid any difficulties in the
  formal manipulations that lead to eq.~(\ref{eqevJ}) and
  eq.(\ref{eqevU}), we consider first a linear equation, with a force
  field a priori determined by solutions to (\ref{eqevJ}) and
  (\ref{eqevU}), and observe that the solutions to this linear
  equation actually solve~(\ref{quenchmaster}). Cf. also ref.\cite{vb}.

Let $\pi_i$  be the continuous function on $\R^N$ defined by
$$ \pi_i(\textbf{V})=\pi_i(v_1,\dots,v_i,\dots,v_N)=v_i, \hspace{0.2in} i=1,\dots,N. $$
From (\ref{coloperator}), we get
\begin{align}\label{calcforJ}
 &\K \pi_i=
\frac{2}{N-1}\sum_{\substack{1\leq j\leq N \\j \neq i}} \int_0^{2\pi}
[\pi_i(R_{ij}(\theta)\textbf{V})-\pi_i(\textbf{V})]
\frac{\mathrm{d\theta}}{2\pi}=-2v_i.
\end{align}
Consider the force field
\begin{equation*}
\widetilde{\textbf{F}}(t)=E\left (\textbf{1}-\frac{\xi(t)}{\bar{u}}\textbf{V} \right ),
\end{equation*}
where $\bar{u}=\widehat{U}_N(0)$ is a constant and $\xi(t)$ satisfies the differential equation
\begin{equation}\label{diffeqforxi}
\frac{\mathrm{d}}{\mathrm{d}t}\xi(t)=E-E\frac{\xi(t)^2}{\bar{u}}-2\xi(t)
\end{equation}
with initial condition $\xi(0)=\widehat{J}_N(0)$. The dynamics of each
particle under the force field $\widetilde{\textbf{F}}$ is given by
\begin{equation*}
\frac{\mathrm{d}}{\mathrm{d}t} v_i(t)=E\left
  (1-\frac{\xi(t)}{\bar{u}}v_i(t) \right ),\hspace{0.3 in}
i=1,\dots,N.
\end{equation*}
or
\begin{equation}\label{dynamic_tilde}
\frac{\mathrm{d}}{\mathrm{d}t} \textbf{V}(t)=\widetilde{\textbf{F}}(t).
\end{equation}
Abbreviating $\gamma(t)=E \xi(t)/ \bar{u}$, we see that
\begin{equation}\label{dynamic_ptilde}
v_i(t)=\alpha_{ts}v_i(s)+\beta_{ts},
\end{equation}
where
\begin{equation*}
\alpha_{ts}=e^{-\int_s^t\gamma(\tau)\mathrm{d}\tau} \hspace{0.2 in} \text{and} \hspace{0.2 in}
\beta_{ts}=E \alpha_{ts} \int_s^t\alpha_{s\tau} \mathrm{d}\tau.
\end{equation*}
Given $\textbf{V}$ at time $s$, let $\widetilde S_{t,s}$ be the flow
such that $\widetilde S_{t,s}(\textbf{V})$ is the unique solution
of~(\ref{dynamic_tilde}) at time $t$ with each component given by
(\ref{dynamic_ptilde}). Let $\widetilde W_N(\textbf{V},t)$ be a
solution to~(\ref{quenchmaster}) with $\widehat{\textbf{F}}_{\widehat
  W_N}(t)=\widetilde{ \textbf{F}}(t)$.
Moreover, let $\widetilde{\mathcal{J}}_{t,s}(\textbf{V})$
be the determinant of the Jacobian of $\widetilde
S_{t,s}(\textbf{V})$, that is,
\begin{equation*}
\widetilde{\mathcal{J}}_{t,s}(\textbf{V})=\left |
  \frac{\mathrm{d}\widetilde
    S_{t,s}(\textbf{V})}{\mathrm{d}\textbf{V}} \right
|=\alpha_{ts}^N.
\end{equation*}
Next, the master equation (\ref{quenchmaster}) with this {\em a priori} determined force field
can be written in mild form as
\begin{equation}
\label{quenchmastertilde}
\frac{1}{\widetilde{\mathcal{J}}_{t,s}}\frac{\mathrm{d}}{\mathrm{d}t}\left(\widetilde
  W_N(\widetilde
  S_{t,s}(\textbf{V}),t)\widetilde{\mathcal{J}}_{t,s}
\right)
=\K\widetilde W_N(\widetilde S_{t,s}(\textbf{V}),t).
\end{equation}
Multiplying both sides of the last equality by $\widetilde{\mathcal{J}}_{t,s}$ and integrating yields
\begin{equation*}
\widetilde W_N(\widetilde
S_{t,s}(\textbf{V}),t)\widetilde{\mathcal{J}}_{t,s} = \widetilde
W_N(\textbf{V},s)
+\int_s^t \K\widetilde W_N(\widetilde
S_{\tau,s}(\textbf{V}),\tau)\widetilde{\mathcal{J}}_{\tau,s}
\mathrm{d}\tau.
\end{equation*}
Now, we multiply both sides of the last equality by $\widetilde
S_{t,s}(v_i)$ and integrate over $\R^N$ with
$\mathrm{d}m^{(N)}=\mathrm{d}v_1\dots \mathrm{d}v_N$. We get
\begin{align*}
&\int_{\R^N}\widetilde W_N(\widetilde S_{t,s}(\textbf{V}),t)
\widetilde S_{t,s}(v_i)\widetilde{\mathcal{J}}_{t,s}
\mathrm{d}m^{(N)}\\
&\hspace{0.15in}=\int_{\R^N}\widetilde W_N(\textbf{V},s)\widetilde S_{t,s}(v_i) \mathrm{d}m^{(N)}+
\int_{\R^N}\widetilde S_{t,s}(v_i) \int_s^t \K\widetilde
W_N(\widetilde
S_{\tau,s}(\textbf{V}),\tau)\widetilde{\mathcal{J}}_{\tau,s}
\mathrm{d}\tau \mathrm{d}m^{(N)}.
\end{align*}
By a change of variables and the property that $\widetilde
S_{t,s}(v_i)=\widetilde S_{t,\tau}(\widetilde S_{\tau,s}(v_i))$ we can
write the last equality as
\begin{align*}
&\int_{\R^N}\widetilde W_N(\textbf{V},t) v_i \mathrm{d}m^{(N)}\\
&\hspace{0.15in}=\int_{\R^N}\widetilde W_N(\textbf{V},s)\widetilde S_{t,s}(v_i) \mathrm{d}m^{(N)}+
\int_s^t \int_{\R^N}\widetilde S_{t,\tau}(v_i)  \K\widetilde
W_N(\textbf{V},\tau) \mathrm{d}m^{(N)} \mathrm{d}\tau.
\end{align*}
Using (\ref{calcforJ}), summing both sides of the last equality from
$i=1$ to $i=N$ and dividing by $N$ leads to the following relation
\begin{align*}
 \widetilde{J}_N(t)=\alpha_{ts}\widetilde{J}_N(s)+\beta_{ts}-2\int_s^t
 \alpha_{t\tau}\widetilde{J}_N(\tau)\mathrm{d}\tau.
\end{align*}
Differentiating the last equality with respect to $t$ yields
\begin{eqnarray*}
  \frac{\mathrm{d}}{\mathrm{d}t}\widetilde{J}_N(t)&=&-\gamma(t)\alpha_{ts}
  \widetilde{J}_N(s)-\gamma(t)\beta_{ts}+E-2\widetilde{J}_N(t)+
  2\int_s^t \gamma(t)\alpha_{t\tau}\widetilde{J}_N(\tau)\mathrm{d}\tau \\
   &=&-\gamma(t)\left (
     \alpha_{ts}\widetilde{J}_N(s)+\beta_{ts}-2\int_s^t
     \alpha_{t\tau}\widetilde{J}_N(\tau)\mathrm{d}\tau \right )+
   E-2\widetilde{J}_N(t)  \\
   &=&-\gamma(t)\widetilde{J}_N(t)+E-2\widetilde{J}_N(t)\\
   &=&E \frac{\xi(t)}{\bar{u}}\widetilde{J}_N(t)+E-2\widetilde{J}_N(t).
\end{eqnarray*}
We now see that $\widetilde{J}_N(t)$ satisfies the differential
equation (\ref{diffeqforxi}) and hence is equal to $\xi(t)$.
A similar calculation also yields
\begin{align*}
 \widetilde{U}_N(t)=\alpha_{ts}^2\widetilde{U}_N(s)+2\alpha_{ts}\beta_{ts}\widetilde{J}_N(s)+\beta_{ts}^2
 -4\int_s^t
 \alpha_{t\tau}\beta_{t\tau}\widetilde{J}_N(\tau)\mathrm{d}\tau.
\end{align*}
Differentiating the last equality with respect to $t$ yields
\begin{align*}
 \frac{\mathrm{d}}{\mathrm{d}t}\widetilde{U}_N(t)=-2E
 \frac{\xi(t)}{\bar{u}}\widetilde{U}_N(t)+2E\widetilde{J}_N(t).
\end{align*}
Using $\xi(t)=\widetilde{J}_N(t)$, we find that
\begin{equation}\label{uU}
\frac{\mathrm{d}}{\mathrm{d}t}\widetilde{U}_N(t)=2E\xi(t)\left(
  1-\frac{\widetilde{U}_N(t)}{\bar{u}}\right),
\end{equation}
and hence $\widetilde{U}_N(t)=\bar{u}$ as the (unique) solution to (\ref{uU}),
\begin{equation*}
\frac{\mathrm{d}}{\mathrm{d}t}\widetilde{U}_N(t)=0.
\end{equation*}
Hence the {\em a posteriori} determined $\widetilde{J}_N(t)$ and
$\widetilde{U}_N(t)$ coincide with $\xi(t)$ and $\bar{u}$, and hence the
$\widetilde{W}_N$ solves~(\ref{quenchmaster}), and the conclusions of
the lemma holds.
\end{proof}
The consequence of the last lemma is that, given initial data, at time
$t$, we can obtain $\widehat{J}_N(t)$ and $\widehat{U}_N(t)$  using
the differential equations (\ref{eqevJ}) and (\ref{eqevU}). This is
independent of knowing $\widehat{W}(\textbf{V},t) $ which is required
when using (\ref{defJhatUhat}) to obtain $\widehat{J}_N(t)$ and
$\widehat{U}_N(t)$. Since $\widehat{U}_N(t)$ is constant in time, in
the remaining of the paper we abbreviate
$$\widehat{U}_N=\widehat{U}_N(0)=\widehat{U}_N(t).$$
Using the new notations, the evolution of each particle is given by
\begin{equation}\label{dynamic_quenchcomp}
\frac{\mathrm{d}}{\mathrm{d}t}\widehat
v_i(t)=E-\frac{E\widehat{J}_N(t) }{\widehat{U}_N}\widehat v_i
(t),\hspace{0.3in} i=1,\cdots,N
\end{equation}
which we also can write as
\begin{equation}\label{dynamic_quench}
\frac{\mathrm{d}}{\mathrm{d}t}\widehat{\textbf{V}}(t)=\widehat{\textbf{F}}_{\widehat W_N}(t).
\end{equation}
Given $\textbf{V}_0$, let $\widehat S_{t,0}$ be the flow such that
$\widehat S_{t,0}(\textbf{V}_0)$ is the unique solution of
(\ref{dynamic_quench}) at time $t$. In what follows we shall need a
bound on a sixth moment of $\widehat W_N$ which is defined as
\begin{equation}
\widehat m_{6,N}(t)=\int_{\R^N} \widehat W_N(\textbf{V},t) v_i^6
\mathrm{d} m^{(N)}, \hspace{0.3in} i=1,\dots,N\,.
\end{equation}
Because $\widehat W_N$ is symmetric, the definition does not depend on
the index $i$.

\begin{lemma}\label{momentconserv}
Assume that $\widehat m_{6,N}(0)<\infty $. For all $t>0$ we have
\begin{equation}
\widehat m_{6,N}(t) \leq C_{\widehat m_{6,N}(0),t}
\end{equation}
where $C_{\widehat m_{6,N}(0),t}$ is a positive constant which depends on $\widehat m_{6,N}(0)$ and $t$.
\end{lemma}
\begin{proof}
Making computations similar to those in the proof of the last lemma we have
\begin{align*}
&\int_{\R^N}\widehat W_N(\textbf{V},t) v_i^6 \mathrm{d}m^{(N)}\\
&\hspace{0.15in}=\int_{\R^N}\widehat W_N(\textbf{V},0)\widehat S_{t,0}(v_i)^6 \mathrm{d}m^{(N)}+
\int_0^t \int_{\R^N}  \K \widehat W_N(\textbf{V},s) \widehat S_{t,s}(v_i)^6 \mathrm{d}m^{(N)} \mathrm{d}s\\
&\hspace{0.15in}=: C_1(t)+ C_2(t).
\end{align*}
The left hand side of the last equality is by definition $\widehat
m_{6,N}(t)$. To estimate $C_1(t)$, we  first note that
$|\widehat{J}_N(t)| \leq \sqrt{\widehat{U}_N}$ by using the Cauchy
Schwartz inequality. Furthermore, a crude estimate on the differential
equation
(\ref{dynamic_quenchcomp}) for the evolution of the particle $\widehat v_i$  yields
\begin{align*}
|\widehat v_i(t)|&=\left|\widehat v_i(0)+\int_0^t
  \left(E-\frac{E\widehat{J}_N(\tau) }{\widehat{U}_N}\widehat v_i
    (\tau) \right)\mathrm{d}\tau \right|\\
&\leq |\widehat v_i(0)|+\int_0^t \left(E+
  E\frac{1}{\sqrt{\widehat{U}_N} }|\widehat v_i
  (\tau)|\right)\mathrm{d}\tau.
\end{align*}
Let
\begin{equation*}
\mu(t)=\frac{Et}{\sqrt{\widehat{U}_N}}.
\end{equation*}
Straightforward estimation yields
\begin{equation*}
\widehat S_{t,0}(v_i)^6 \leq 16e^{8\mu(t)} \left(v_i^6+\widehat{U}_N^3 \right).
\end{equation*}
Hence,
\begin{equation}
C_1(t)\leq 16e^{8 \mu(t)}\left (\widehat m_{6,N}(0)+\widehat{U}_N^3 \right).
\end{equation}
To estimate $C_2(t)$, using that $\K$ is self-adjoint, we can write
\begin{equation*}
C_2(t)=\int_0^t \int_{\R^N}  \widehat W_N(\textbf{V},s)\K \widehat
S_{t,s}(v_i)^6 \mathrm{d}m^{(N)} \mathrm{d}s.
\end{equation*}
A calculation similar to (\ref{calcforJ}) on  $\K \widehat
S_{t,s}(v_i)^6$ and the inequality \\
$(a+b)^2\leq 2(a^2+b^2)$ for $a,b \in \R$ yield
\begin{equation*}
C_2(t) \leq 64 \int_0^t \int_{\R^N}   \widehat W_N(\textbf{V},s)
\widehat S_{t,s}(v_i)^6 \mathrm{d}m^{(N)} \mathrm{d}s.
\end{equation*}
Since $\widehat{S}_{t,s}(v_i)$ is of the form
$\widehat{S}_{t,s}(v_i)=\widehat{\alpha}_{ts}v_i+\widehat{\beta}_{ts}$,
it follows that
\begin{equation*}
C_2(t) \leq C \left( A_1 \int_0^t \widehat m_{6,N}(s)\mathrm{d}s +A_2\right ),
\end{equation*}
where $C$ is a positive constant and
\begin{equation*}
A_1=\sup_{0\leq s\leq t}\widehat{\alpha}_{ts}^6, \hspace{0.2in}
A_2=\sup_{0\leq s\leq t}\widehat{\beta}_{ts}^6.
\end{equation*}
Combining the estimates above, we have
\begin{equation}
\widehat m_{6,N}(t)\leq 16e^{8\mu(t)}\left (\widehat m_{6,N}(0)+\widehat{U}_N^3\right) +
C \left( A_1 \int_0^t \widehat m_{6,N}(s)\mathrm{d}s +A_2\right ).
\end{equation}
Rewriting this inequality slightly we can apply Gronwall's lemma to obtain
\begin{equation}
\widehat m_{6,N}(t)\leq C_{\widehat m_{6,N}(0),t}.
\end{equation}
where $C_{\widehat m_{6,N}(0),t}$ is constant depending on $\widehat
m_{6,N}(0)$ and $t$. Note that the Gronwall's lemma gives that
$C_{\widehat m_{6,N}(0),t}$ depends on $N$ only from the initial condition $m_{6,N}(0)$.  \\
\end{proof}
\section{Propagation of chaos for the quenched master equation}
\noindent In the previous section we defined the quenched master equation (\ref{quenchmaster}). The goal of this section is to show that it propagates chaos.
However, in order to take advantage of this to show that the master equation (\ref{Nparticlemaster}) propagates chaos, we need to know at which rate
(\ref{quenchmaster}) propagates chaos. The reason that propagation of chaos holds for (\ref{quenchmaster}) is that the particles
driven by the quenched force field evolve independently between collisions. In~\cite{kac} it is shown that given chaotic initial data the master equation (\ref{quenchmaster}) without the term $\nabla\cdot(\widehat{\textbf{F}}_{\widehat W_N}(t)\widehat W_N(\textbf{V},t))$ propagates chaos.
The main idea in the proof of Kac is that as $N$ tends to infinity,
the probability that any given particle collides with some other particle more than once tends to zero. By isolating the
contribution of ``recollisions'' to the evolution, and showing that their contribution is negligible in the limit, Kac deduced his
asymptotic factorization property. To state this precisely, and to state our quantitative version, we first introduce some
notation, defining the marginals of $\widehat W_N(\textbf{V},t)$. Let
\begin{equation*}
\widehat{f}_1^N(v_1,t)=\int_{\R^{N-1}} \widehat W_N(\textbf{V},t)\mathrm{d}m^{(N-1)}
\end{equation*}
be the one-particle marginal of $\widehat W_N(\textbf{V},t)$ at time $t$. Since $\widehat W_N$ is symmetric under permutation of the variables $v_1,\dots,v_n$, it does not matter which variables we integrate over. Similarly, the $k$-th marginal of $\widehat W_N(\textbf{V},t)$ at time $t$ is defined as
\begin{equation}\label{kmarg-quench}
\widehat{f}_k^N(v_1,\dots,v_k ,t)=\int_{\R^{N-k}} \widehat W_N(\textbf{V},t)\mathrm{d}m^{(N-k)}.
\end{equation}
The qualitative result of Kac is that
\[
\lim_{N\to\infty} \int_{\R^2} \left[ \widehat{f}_2^N(v_1,v_2,t) -
  \widehat{f}_1^N(v_1,t)\widehat{f}_1^N(v_,t)\right]\phi(v_1,v_2){\rm
  d}m^{(2)} = 0\
\]
for all bounded, continuous functions $\phi$.\\
\indent The main result of this section is:
\begin{theorem}\label{thm_quant_thermostat}
Let $\{\widehat W_N(\textbf{V},0)\}_{N\in \N}$ be a sequence of
symmetric probability densities on $\R^N$ such that
\begin{align}\label{intialdataassumption}
\nonumber &\int_{\R^k} \widehat f_k^N(v_1,\dots,v_k,0)\phi(v_1,\dots,v_k)\mathrm{d}m^{(k)}\\
&\hspace{0.4in} =\int_{\R^k} \widehat f_1^N(v_1,0) \cdots
f_1^N(v_k,0)\phi(v_1,\dots,v_k)\mathrm{d}m^{(k)}+ R_{0,N},
\end{align}
for all $k\in\N$, where $\phi(v_1,\dots,v_k)$ is a
bounded continuous function on $\R^k$ and
\begin{equation*}
R_{0,N}\leq C_0\frac{k}{N}||\phi||_{\infty}\,,
\end{equation*}
with $C_0$ being a positive constant. Then, we have for $0\leq t \leq T$, where $T<\infty$ that
\begin{align}
\nonumber &\int_{\R^k} \widehat f_k^N(v_1,\dots,v_k,t)\phi(v_1,\dots,v_k)\mathrm{d}m^{(k)}\\
&\hspace{0.3in}=\int_{\R^k} \widehat f_1^N(v_1,t) \cdots
f_1^N(v_k,t)\phi(v_1,\dots,v_k)\mathrm{d}m^{(k)}+ R_{T,N},
\end{align}
where
\begin{equation*}
R_{T,N}\leq C(T)\frac{k}{N}||\phi||_{\infty}
\end{equation*}
and $C(T)$ is constant depending only on $T$.
\end{theorem}
This result provides the means to adapt the strategy developed
in~\cite{bcelm} for controlling the effects of correlations
that are introduced by the thermostatting force
field. In~\cite{bcelm}, the collisions were not binary collisions, but
were a model of collisions
with background scatterers. These collisions did not introduce any
correlations at all, and in that work the analogous quenched
process exactly propagated {\em independence}. This facilitated appeal
to the Law of Large Numbers. When we need to apply the Law of Large Numbers here,
the individual velocities in our quenched process will
not be independent, and we must
quantify the lack of independence. Theorem~\ref{thm_quant_thermostat}
provides the means to do this, and may be of independent interest.
The proof Theorem~\ref{thm_quant_thermostat} is build on ideas from the original proof of
~\cite{kac}, in particular, on his idea of controlling the effect of
recollisions, and the refined combinatoric arguments
from~\cite{BCP}. However,
the proof  is rather long. We divide it into $5$ steps.

\begin{proof}
\underline{\textbf{STEP 1}}\\
In this step we express the solution $\widehat W_N(\textbf{V},t)$ of
(\ref{quenchmaster}) as a series depending on $\widehat
W_N(\textbf{V},0)$ and use this to find an expression for the k-th
marginal of $\widehat W_N(\textbf{V},t)$ at time $t$. Recall that the
quenched master equation is given by
\begin{equation}
\frac{\partial}{\partial t}\widehat
W_N+\nabla\cdot(\widehat{\textbf{F}}_{\widehat W_N}\widehat W_N)=\K
\widehat W_N,
\end{equation}
with initial data $\widehat W_N(\textbf{V},0)=\widehat W_{N,0}(\textbf{V})$, and where
$\widehat{\textbf{F}}_{\widehat W_N} $ is given  by
(\ref{modffield}).
Let $\widehat P_{t,0}\widehat W_N(\textbf{V},0)$ denote the solution to its homogenous part where
the operator $\widehat P_{t,s}:L^1\rightarrow L^1$ transforms the
density $\widehat W_N$ from time $s$ to time $t$. Explicitly
\begin{equation}
\widehat P_{t,s} \widehat W_N(\textbf{V},s)= \widehat W_N(\widehat
S_{t,s}^{-1}(\textbf{V}),s) \widehat{\mathcal{J}}_{t,s}^{-1},
\end{equation}
where $\widehat{\mathcal{J}}_{t,s}$ is the determinant of
the Jacobian of $\widehat S_{t,s}(\textbf{V})$, i.e.,
$$\widehat{\mathcal{J}}_{t,s}=\left |\frac{\mathrm{d}\widehat
    S_{t,s}(\textbf{V})}{\mathrm{d}\textbf{V}}\right |.$$
By the Duhamel formula,
\begin{equation}\label{solinhomdiff}
\widehat W_N(\textbf{V},t)=\widehat P_{t,0}\widehat W_N(\textbf{V},0)
+\int_0^{t} \widehat P_{t,s} \K \widehat W_N(\textbf{V},s)\mathrm{d}s.
\end{equation}
Iterating (\ref{solinhomdiff}) expresses $\widehat W_N(\textbf{V},t)$  as a series:
\begin{equation}\label{iterationsol}
\widehat W_N(\textbf{V},t)=\widehat P_{t,0}\widehat W_N(\textbf{V},0)+
\sum_{j=1}^{\infty}\int_{A_j} \widehat P_{t,t_j}\left(\prod_{i=0}^{j-1}\K \widehat P_{t_{j-i},t_{j-i-1}}\right)
\widehat W_N(\textbf{V},0)\mathrm{d}t_1 \dots \mathrm{d}t_j
\end{equation}
where
$$A_j=\{0=t_0<t_1<t_2,\hspace{0.05in}\dots \hspace{0.05in} , 0<t_j<t \}.$$
%
%
%
For a continuous function $\phi$ of $k$ variables $v_1,\dots,v_k$, it
follows from (\ref{kmarg-quench}) and (\ref{iterationsol}) that
\begin{align}\label{differentfofmarginal}
 \nonumber   &\int_{\R^k}
    \widehat
    f_k^N(v_1,\dots,v_k,t)\phi(v_1,\dots,v_k)\mathrm{d}m^{(k)}=\int_{\R^N}
    \widehat W_N(\textbf{V},t)\phi(v_1,\dots,v_k)\mathrm{d}m^{(N)}=\\
 \nonumber   &=\int_{\R^N}\widehat W_N(\textbf{V},0)\widehat
 P_{t,0}^{*}\phi \hspace{0.02in}\mathrm{d}m^{(N)}\\
    &\hspace{0.3in}+\sum_{j=1}^{\infty}\int_{A_j} \int_{\R^N} \widehat W_N(\textbf{V},0)\left(\prod_{i=0}^{j-1}
    \K \widehat P_{t_{j-i},t_{j-i-1}}\right)^{*} \widehat
  P_{t,t_j}^{*} \phi \hspace{0.02in} \mathrm{d}m^{(N)} \mathrm{d}t_1
  \dots \mathrm{d}t_j.
\end{align}
The operator $\widehat P_{t,s}^*:L^\infty \rightarrow L^\infty$ is the
adjoint of the operator $\widehat P_{t,s}$. Explicitly,
$$\widehat P_{t,s}^*\phi(v_1,\dots,v_k)= \phi(\widehat S_{t,s}(v_1), \dots , \widehat S_{t,s}(v_k) ),$$
and $\widehat P_{t,s}^*\phi$ is still a function of $v_1,\dots, v_k$
but also depends on $\widehat{J}_N(t)$ and $\widehat{U}_N$. Moreover,
the operator $\widehat P_{t,s}^*$ preserves the $L^{\infty}$ norm,
i.e., $||\widehat P_{t,s}^*\phi||_\infty =||\phi||_\infty$.
\\ \\
\noindent \underline{\textbf{STEP 2}}\\
Observing what we obtained in (\ref{differentfofmarginal}), we now need to see how the operator
\begin{equation}\label{powersofg}
\left(\prod_{i=0}^{j-1}
    \K \widehat P_{t_{j-i},t_{j-i-1}}\right)^{*}\widehat P_{t,t_j}^{*}
\end{equation}
acts on a bounded continuous function $\phi$ depending on finitely
many variables. Introducing the notation
$$\Gamma_{t,t_1}^1=\K \widehat
P^*_{t,t_1},\hspace{0.2in}\Gamma_{t,t_2,t_1}^2=\K \widehat
P^*_{t_2,t_1} \K \widehat P^*_{t,t_2},\hspace{0.2in} \dots $$
we can now write (\ref{powersofg}) as
\begin{equation}\label{notgamma}
\widehat
P^*_{t_1,0}\Gamma^{j}_{t,t_j,\dots,t_1}=\left(\prod_{i=0}^{j-1}\K
  \widehat P_{t_{j-i},t_{j-i-1}}\right)^{*}\widehat P_{t,t_j}^{*}.
\end{equation}
Following~\cite{kac}, let us see how  $\widehat
P^*_{t_1,0}\Gamma^{j}_{t,t_j,\dots,t_1}$ acts on a function
$\phi_1(v_1)$ depending only on one variable. For $j=1$, we have
\begin{align*}
&\widehat P^*_{t_1,0}\Gamma^1_{t,t_1}
\phi_1=\frac{2}{N-1}\sum_{j=2}^{N}\widehat
P^*_{t_1,0}(Q_{(1,j)}-I)\widehat P^*_{t,t_1}\phi_1(v_1)\\
&=\frac{2}{N-1}\sum_{j=2}^{N}\widehat P^*_{t_1,0}(Q_{(1,j)}-I)\phi_{1;1}(v_1;t,t_1), \\
\end{align*}
where $\phi_{1;1}(v_1;t,t_1)=\widehat P^*_{t,t_1}\phi_1(v_1)$. The
operator $Q$ adds a new variable $v_j$ to $\phi_{1;1}(v_1;t,t_1)$ at
time $t_1$. Setting
\begin{align*}
\phi_{2;1}(v_1,v_2;t,t_1)=2(Q_{(1,2)}-I)\phi_{1;1}(v_1;t,t_1),
\end{align*}
we have
\begin{equation*}
\widehat P^*_{t_1,0}\Gamma^1_{t,t_1}
\phi_1=\frac{1}{N-1}\sum_{j=2}^{N}\widehat P^*_{t_1,0}
\phi_{2;1}(v_1,v_j;t,t_1).
\end{equation*}
For $j=2$ we get
\begin{align*}
\widehat P^*_{t_1,0}\Gamma^2_{t,t_2,t_1}
\phi_1=\frac{1}{N-1}\sum_{j=2}^{N}\widehat
P^*_{t_1,0}\Gamma_{t_2,t_1}^1 \phi_{2;1}(v_1,v_j;t,t_2).
\end{align*}
Setting $\phi_{2;2}(v_1,v_2;t,t_2,t_1)= \widehat P^*_{t_2,t_1}\phi_{2;1}(v_1,v_2;t,t_2)$, we see that
\begin{align}\label{calccolcol}
\nonumber &\widehat P^*_{t_1,0}\Gamma_{t_2,t_1}^1 \phi_{2;1}(v_1,v_2;t,t_2)\\
\nonumber &=\frac{2}{N-1}\widehat P^*_{t_1,0}(Q_{(1,2)}-I)\phi_{2;2}(v_1,v_2;t,t_2,t_1)\\
\nonumber &+\frac{2}{N-1}\sum_{j=3}^{N}\widehat P^*_{t_1,0}(Q_{(1,j)}-I)\phi_{2;2}(v_1,v_2;t,t_2,t_1)\\
 &+\frac{2}{N-1}\sum_{j=3}^{N}\widehat P^*_{t_1,0}(Q_{(2,j)}-I)\phi_{2;2}(v_1,v_2;t,t_2,t_1).
\end{align}
When $Q$ acts on $\phi_{2;2}$ in two last expressions above again a
new  velocity variable is created in $\phi_{2;2}$ leading to
$\phi_{3;2}$. In this fashion each time $\Gamma^1$ acts on a function,
a  new time variable and a new velocity variable is created .
Since $\widehat P^*_{t,s}$ preserves the $L^{\infty}$ norm , it
follows that $||\phi_{2;2}||_{\infty}\leq 4||\phi_1 ||_{\infty}$. This
in turn implies that
\begin{equation*}
||\widehat P^*_{t_1,0}\Gamma^1_{t,t_1} \phi_1||_{\infty} \leq \hspace{0.05in} 4||\phi_1||_{\infty}.
\end{equation*}
From (\ref{calccolcol}) it also follows that
\begin{align*}
||\widehat P^*_{t_1,0}\Gamma^2_{t,t_2,t_1} \phi_1||_{\infty} \leq
\left(\frac{2}{N-1}+\frac{4(N-2)}{N-1}\right)8||\phi_1||_{\infty}\leq
2!\hspace{0.05in} 4^2||\phi_1||_{\infty}.
\end{align*}
It is tedious but straightforward to show that in general we have
\begin{equation}\label{bound_series_terms}
||\widehat P^*_{t_1,0}\Gamma^j_{t,t_j,\dots,t_1} \phi_1||_{\infty}
\leq j! \hspace{0.05in}  4^j ||\phi_1||_{\infty}.
\end{equation}
A detailed proof of (\ref{bound_series_terms}) in the
case $\widehat P^*_{t_1,0}\Gamma^j_{t,t_j,\dots,t_1}=\K^j$, i.e, $\widehat P^*_{t,s}=Id$
for all $t$ and $s$ can be found in~\cite{BCP}, for $j+1<N$ as well as for $j+1\ge N$ (the
latter has to be handled a little differently). In our
case the proof follows along the same lines since $\widehat P^*_{t,s}$
preserves the $L^\infty$ norm. More generally, if $\phi_m$ is function
of $m$ variables, $m\geq2 $, it can be shown by induction that
\begin{eqnarray}\label{estimateg}
  ||\widehat P^*_{t_1,0}\Gamma^j_{t,t_j,\dots,t_1} \phi_m||_{\infty}
  &\leq& 4^j j! \left( \begin{array}{c} m+j-1 \\ j \end{array}
  \right)||\phi_m||_{\infty}. 
\end{eqnarray}
An important feature of the estimate (\ref{estimateg}) is that it is
independent of $N$, and we note that in (\ref{calccolcol}) for large
$N$ the first term is small but the two last terms where new velocity
variables are added have an impact.
\\ \\
\underline{\textbf{STEP 3}}\\
In the  previous step we found that the action of
$P^*_{t_1,0}\Gamma_{t,t_1}$ on $\phi_1$ results in a sum in which the
terms consist of  functions $\phi_{2;1}$ depending on two velocity
variables by adding a velocity variable and a time variable to
$\phi_1$, and $P^*_{t_1,0}\Gamma^2_{t,t_2,t_1}$ acting on $\phi_1$
results in a sum in which the terms consist of functions $\phi_{3;2}$
depending on three velocity variables  by adding two velocity
variables and two time variables to~$\phi_1$. In this step we look at
the action of $\widehat P^*_{t_1,0}\Gamma^1_{t_2,t_1}$ on a
function $\phi_{k;l-1}$ depending on $k$ velocity variables and $l-1$
time variables to see that for large $N$, only the terms where a new
velocity variable is added make a significant contribution. To be more
precise, we have the following lemma which corresponds to Lemma $3.5$
in~\cite{BCP} which deals with the case $\widehat
P^*_{t_1,0}\Gamma^j_{t,t_j,\dots,t_1}=\K^j$.
\begin{lemma}
\label{lem_3.2}
Assume that $\widehat{W}_N(\textbf{V},0)$ is a symmetric probability density on $\R^N$.
For $l\geq 2$, let $\phi_{k;l-1}=\phi_{k;l-1}(v_1,\dots,v_k;t,t_l,\dots,t_2)$, where
$\phi_{k;0}=\phi_k(v_1,\dots,v_k)$ is a bounded continuous function of the variables $v_1\dots,v_k$.
Define $\widehat P^*_{t_1,0}\phi_{k+1;l}$ as
\begin{align}\label{formulagplus}
\nonumber &\widehat P^*_{t_1,0}\phi_{k+1;l}(v_1,\dots,v_{k+1};t,t_l,\dots,t_1)\\
&\hspace{0.3in}=2\sum_{i=1}^{k}\widehat P^*_{t_1,0}(Q_{(i,k+1)}-I)\widehat P^*_{t_2,t_1}\phi_{k;l-1}(v_1,\dots,v_i,\dots,v_k;t,t_l,\dots,t_2).
\end{align}
Then, we have
\begin{align*}
&\int_{\R^N} \widehat{W}_N(\textbf{V},0)\widehat
P^*_{t_1,0}\Gamma^1_{t_2,t_1} \phi_{k;l-1}
\hspace{0.05in}\mathrm{d}m^{(N)}\\
&\hspace{0.1in}=\int_{\R^N} \widehat{W}_N(\textbf{V},0) \widehat
P^*_{t_1,0} \big
[\phi_{k+1;l}(v_1,\dots,v_{k+1};t,t_l,\dots,t_2,t_1)\\
&\hspace{1.5in}+\phi^R_{k+1;l}(v_1, \dots ,v_{k+1};t,t_l,\dots,t_2,t_1)\big ] \mathrm{d}m^{(N)}.
\end{align*}
where $\phi^R_{k+1;l}$ is a function depending $v_1,\dots,v_{k+1}$ and $t_1,\dots,t_l$, and
\begin{align}\label{restestimate}
||&\phi^R_{k+1;l}||_\infty \leq \frac{4}{N-1}\left(
  k(k-1)+\frac{k(k-1)}{2} \right ) ||\phi_{k;l-1}||_\infty= 6
\frac{k(k-1)}{(N-1)}||\phi_{k;l-1}||_\infty.
\end{align}
\end{lemma}
\begin{proof}
By the definition of $\K$ and since $\phi_{k;l-1}$ depends on $v_1,\dots,v_k$, we have
\begin{align*}
&\int_{\R^N} \widehat{W}_N(\textbf{V},0)\widehat
P^*_{t_1,0}\Gamma^1_{t_2,t_1}
\phi_{k;l-1}(v_1,\dots,v_k;t,t_l,\dots,t_2)\hspace{0.05in}\mathrm{d}m^{(N)}=\\
&\frac{2}{N-1}\sum_{i=1}^{k}\sum_{j=i+1}^{N}\int_{\R^N}
\widehat{W}_N(\textbf{V},0)\widehat P^*_{t_1,0}(Q_{(i,j)}-I)\widehat
P^*_{t_2,t_1}\phi_{k;l-1} \hspace{0.05in} \mathrm{d}m^{(N)}\\
&=\frac{2}{N-1}\sum_{i=1}^{k}\sum_{j=k+1}^{N}\int_{\R^N}
\widehat{W}_N(\textbf{V},s)\widehat P^*_{t_1,0}(Q_{(i,j)}-I)\widehat
P^*_{t_2,t_1}\phi_{k;l-1} \hspace{0.05in} \mathrm{d}m^{(N)}\\
&\hspace{0.1in}+\frac{2}{N-1}\sum_{i=1}^{k}\sum_{j=i+1}^{k}\int_{\R^N}
\widehat{W}_N(\textbf{V},0)\widehat P^*_{t_1,0}(Q_{(i,j)}-I)
\widehat P^*_{t_2,t_1}\phi_{k;l-1}\hspace{0.05in} \mathrm{d}m^{(N)}.
\end{align*}
Because $ \widehat{W}_N(\textbf{V},0)$ is a symmetric probability density, we find that
\begin{align*}
&\frac{2}{N-1}\sum_{i=1}^{k}\sum_{j=k+1}^{N}\int_{\R^N}
\widehat{W}_N(\textbf{V},0)\widehat
P^*_{t_1,0}(Q_{(i,j)}-I)P^*_{t_2,t_1}\phi_{k;l-1} \hspace{0.05in}
\mathrm{d}m^{(N)}\\
&\hspace{0.2in}=2\frac{N-k}{N-1}\sum_{i=1}^{k}\int_{\R^N}
\widehat{W}_N(\textbf{V},0)\widehat P^*_{t_1,0}(Q_{(i,k+1)}-I)\widehat
P^*_{t_2,t_1}\phi_{k;l-1} \hspace{0.05in} \mathrm{d}m^{(N)}.
\end{align*}
Using this, it follows that
\begin{align}\label{formeulanextcol}
\nonumber &\int_{\R^N}\widehat{W}_N(\textbf{V},0)\widehat
P^*_{t_1,0}\Gamma^1_{t_2,t_1}
\phi_{k;l-1}\hspace{0.05in}\mathrm{d}m^{(N)}=\\
\nonumber
&=2\frac{N-k}{N-1}\sum_{i=1}^{k}\int_{\R^N}\widehat{W}_N(\textbf{V},0)\widehat
P^*_{t_1,0}(Q_{(i,k+1)}-I)\widehat P^*_{t_2,t_1}\phi_{k;l-1}
\hspace{0.05in} \mathrm{d}m^{(N)}\\
\nonumber
&+\frac{2}{N-1}\sum_{i=1}^{k}\sum_{j=i+1}^{k}\int_{\R^N}\widehat{W}_N(\textbf{V},0)\widehat
P^*_{t_1,0}(Q_{(i,j)}-I)\widehat P^*_{t_2,t_1}\phi_{k;l-1}
\hspace{0.05in} \mathrm{d}m^{(N)}\\
\nonumber &=: \int_{\R^N}\widehat{W}_N(\textbf{V},0)\widehat
P^*_{t_1,0}\big [\phi_{k+1;l}(v_1,\dots,v_{k+1};t,t_l,\dots,t_1)+\\
&\hspace{1in}+\phi^R_{k+1;l}(v_1,\dots,v_{k+1};t,t_l,\dots,t_1)\big ]\mathrm{d}m^{(N)},
\end{align}
where
\begin{align*}
&\phi^R_{k+1;l}(v_1,\dots,v_{k+1};t,t_l,\dots,t_1,0)\\
&=2\frac{1-k}{N-1}\sum_{i=1}^{k}\widehat
P^*_{t_1,0}(Q_{(i,k+1)}-I)\widehat P^*_{t_2,t_1}\phi_{k;l-1}+\\
&\hspace{0.5in}+\frac{2}{N-1}\sum_{i=1}^{k}\sum_{j=i+1}^{k}\widehat
P^*_{t_1,0}(Q_{(i,j)}-I)\widehat P^*_{t_2,t_1}\phi_{k;l-1},
\end{align*}
and finally we obtain the estimate
\begin{align}\label{restestimate}
||&\phi^R_{k+1;l}||_\infty \leq \frac{4}{N-1}\left(
  k(k-1)+\frac{k(k-1)}{2} \right ) ||\phi_{k;l-1}||_\infty= 6
\frac{k(k-1)}{(N-1)}||\phi_{k;l-1}||_\infty.
\end{align}
\end{proof}
The construction in step 2 and 3 shows that $\widehat
P^*_{t_1,0}\phi_{k+1;l}(v_1,\dots,v_{k+1};t,t_l,\dots,t_1)$, up
to an error term that vanishes in the limit of large $N$ like $1/N$, can be
written as a sum of terms of which each can be represented by a binary tree,
which determines in which order new velocities are added.
For example, starting with one velocity $v_1$ at time $t$ and adding
three new velocities as described above, we could find the following
sequence of graphs:



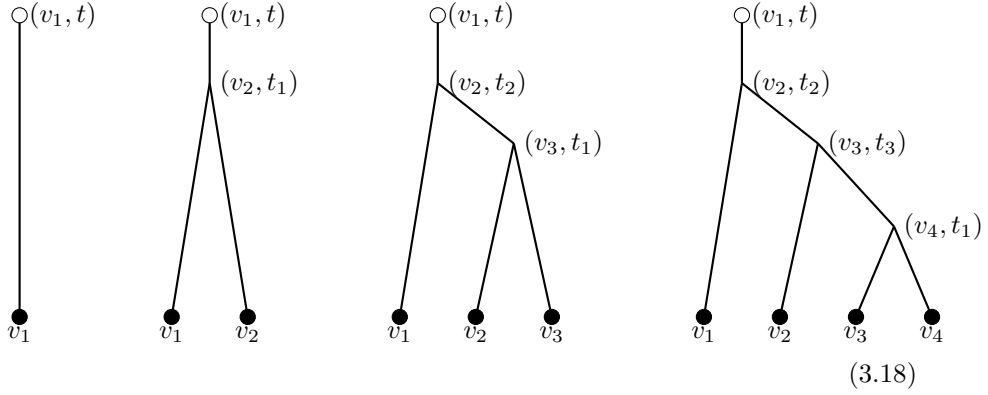
\begin{dispmathenv}
\label{eq:tree1}
\begin{tikzpicture}
\draw[black,fill] (0,0) circle (0.1)  node[below] {$v_1$};
\draw[black] (0,4) circle (0.1) node[right] {$(v_1,t)$};
\draw[black,thick] (0,0) -- (0, 3.9);
\draw[black,fill] (2,0) circle (0.1) node[below] {$v_1$};
\draw[black,fill] (3,0) circle (0.1) node[below] {$v_2$};
\draw[black] (2.5,4) circle (0.1) node[right] {$(v_1,t)$};
\draw[black,thick] (2.5,3.1) node[right] {$(v_2,t_1)$} -- (2.5, 3.9) ;
\draw[black,thick] (2,0) -- (2.5,3.1);
\draw[black,thick] (3,0) -- (2.5,3.1);
\draw[black,fill] (5,0) circle (0.1)  node[below] {$v_1$};
\draw[black,fill] (6,0) circle (0.1)  node[below] {$v_2$};
\draw[black,fill] (7,0) circle (0.1)  node[below] {$v_3$};
\draw[black] (5.5,4) circle (0.1) node[right] {$(v_1,t)$};
\draw[black,thick] (5.5,3.1) node[right] {$(v_2,t_2)$} -- (5.5, 3.9) ;
\draw[black,thick] (5,0) -- (5.5,3.1);
\draw[black,thick] (6.5,2.3) node[right] {$(v_3,t_1)$} -- (5.5, 3.1) ;
\draw[black,thick] (6,0) -- (6.5,2.3);
\draw[black,thick] (7,0) -- (6.5,2.3);
\draw[black,fill] (9,0) circle (0.1)  node[below] {$v_1$};
\draw[black,fill] (10,0) circle (0.1)  node[below] {$v_2$};
\draw[black,fill] (11,0) circle (0.1)  node[below] {$v_3$};
\draw[black,fill] (12,0) circle (0.1)  node[below] {$v_4$};
\draw[black] (9.5,4) circle (0.1) node[right] {$(v_1,t)$};
\draw[black,thick] (9.5,3.1) node[right] {$(v_2,t_2)$} -- (9.5, 3.9) ;
\draw[black,thick] (9,0) -- (9.5,3.1);
\draw[black,thick] (10.5,2.3) node[right] {$(v_3,t_3)$} -- (9.5, 3.1) ;
\draw[black,thick] (10,0) -- (10.5,2.3);
\draw[black,thick] (11.5,1.2) node[right] {$(v_4,t_1)$}  -- (10.5,2.3);
\draw[black,thick] (11,0) -- (11.5,1.2);
\draw[black,thick] (12,0) -- (11.5,1.2);
\end{tikzpicture}
\end{dispmathenv}
where the new velocities are always added to the right branch of the tree,
or
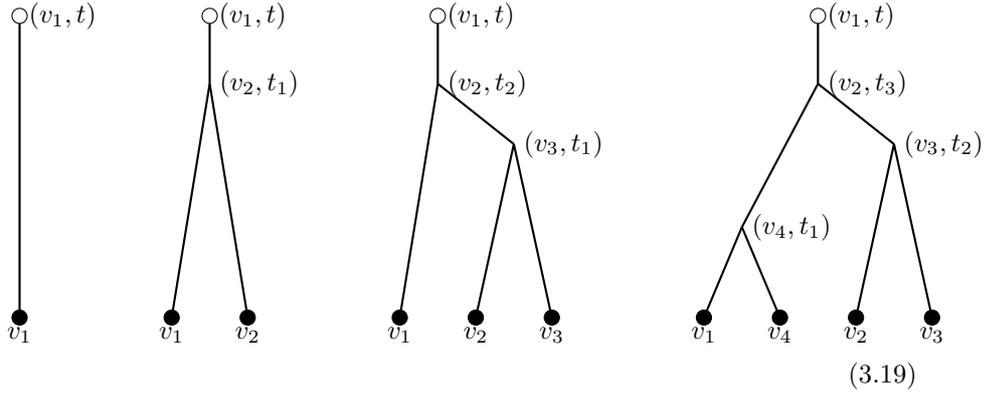
\begin{dispmathenv}
\label{eq:tree2}
\begin{tikzpicture}
\draw[black,fill] (0,0) circle (0.1)  node[below] {$v_1$};
\draw[black] (0,4) circle (0.1) node[right] {$(v_1,t)$};
\draw[black,thick] (0,0) -- (0, 3.9);
\draw[black,fill] (2,0) circle (0.1) node[below] {$v_1$};
\draw[black,fill] (3,0) circle (0.1) node[below] {$v_2$};
\draw[black] (2.5,4) circle (0.1) node[right] {$(v_1,t)$};
\draw[black,thick] (2.5,3.1) node[right] {$(v_2,t_1)$} -- (2.5, 3.9) ;
\draw[black,thick] (2,0) -- (2.5,3.1);
\draw[black,thick] (3,0) -- (2.5,3.1);
\draw[black,fill] (5,0) circle (0.1)  node[below] {$v_1$};
\draw[black,fill] (6,0) circle (0.1)  node[below] {$v_2$};
\draw[black,fill] (7,0) circle (0.1)  node[below] {$v_3$};
\draw[black] (5.5,4) circle (0.1) node[right] {$(v_1,t)$};
\draw[black,thick] (5.5,3.1) node[right] {$(v_2,t_2)$} -- (5.5, 3.9) ;
\draw[black,thick] (5,0) -- (5.5,3.1);
\draw[black,thick] (6.5,2.3) node[right] {$(v_3,t_1)$} -- (5.5, 3.1) ;
\draw[black,thick] (6,0) -- (6.5,2.3);
\draw[black,thick] (7,0) -- (6.5,2.3);
\draw[black,fill] (9,0) circle (0.1)  node[below] {$v_1$};
\draw[black,fill] (10,0) circle (0.1)  node[below] {$v_4$};
\draw[black,fill] (11,0) circle (0.1)  node[below] {$v_2$};
\draw[black,fill] (12,0) circle (0.1)  node[below] {$v_3$};
\draw[black] (10.5,4) circle (0.1) node[right] {$(v_1,t)$};
\draw[black,thick] (10.5,3.1) node[right] {$(v_2,t_3)$} -- (10.5, 3.9) ;
\draw[black,thick] (11.5,2.3) node[right] {$(v_3,t_2)$} -- (10.5, 3.1) ;
\draw[black,thick] (11,0) -- (11.5,2.3);
\draw[black,thick] (12,0) -- (11.5,2.3);
\draw[black,thick] (9.5,1.2) node[right] {$(v_4,t_1)$} -- (10.5, 3.1) ;
\draw[black,thick] (9,0) -- (9.5,1.2);
\draw[black,thick] (10,0) -- (9.5,1.2);
\end{tikzpicture}
\end{dispmathenv}
where the tree is built symmetrically. The terms of order $1/N$ that are
deferred to the rest term can be represented by trees very much in the same way,
but may have two or more leafs with the same velocity variable. This
is exactly as in the original Kac paper as far as the collisions go:
the collision process is independent of the force field and of the
state of the $N$-particle system, and hence the number of terms
represented by a particular tree, and the distribution of time points
where new velocities are added (giving a new branch of the tree) are
exactly the same in our setting as in the original one.
But what happens between the collisions is important, and hence the
added velocities are noted together with the time of addition.
In this construction the velocities are added in increasing order of
indices, but it is important to understand that any set of four different variables out of
$v_1,...,v_N$, or any permuation of $v_1,...,v_4$ would give the same result.
\\ \\
\noindent \underline{\textbf{STEP 4}}\\
In (\ref{differentfofmarginal}) we obtained an expression for the
$k$-th marginal $f_k^N$ of $\widehat W_N(\textbf{V},t)$ at time $t$ as
a series. In this step we check that this series representation is
uniformly convergent in $N$. We will only consider the case $k=1,2$,
the other cases being similar but more tedious.
Setting $\phi_{1;0}(v_1)=\phi_1(v_1)$ and defining
$P^*_{t_1,0}\phi_{j+1;j}$ inductively by (\ref{formulagplus}) we have
that (\ref{differentfofmarginal}) without the first term equals
(recall also notation (\ref{notgamma}))
\begin{align}\label{seriesnew}
\nonumber &\sum_{j=1}^{\infty}\int_{\R^N}\int_{A_j}\widehat
W_N(\textbf{V},0)\widehat P^*_{t_1,0}\Gamma^j_{t,t_j,\dots,t_1}
\phi_{1;0}\mathrm{d}m^{(N)}\mathrm{d}t_1 \dots \mathrm{d}t_j \\
\nonumber &=\sum_{j=1}^{\infty} \int_{\R^N}\int_{A_j}\widehat
W_N(\textbf{V},0)\widehat
P^*_{t_1,0}\phi_{j+1;j}(v_1,\dots,v_{j+1};t,t_j,\dots,t_1)
\mathrm{d}m^{(N)}\mathrm{d}t_1 \dots \mathrm{d}t_j\\
\nonumber &\hspace{0.2in}+ \sum_{j=1}^{\infty}
\int_{\R^N}\int_{A_j}\widehat W_N(\textbf{V},0)\widehat
P^*_{t_1,0}\left(\sum_{i=1}^{j}
\Gamma^{j-i}_{t_{j-(i-1)},\dots,t_1} \phi^R_{i+1;i}(v_1,\dots,v_{i+1};t,\dots,t_{j-(i-1)}) \right)\\
&\hspace{2.3in} \mathrm{d}m^{(N)}\mathrm{d}t_1 \dots \mathrm{d}t_j.
\end{align}
By induction on (\ref{formulagplus}) it follows that
\begin{align*}
||\phi_{j+1;j}||_{\infty}\leq 4^j \hspace{0.015in} j! \hspace{0.015in} ||\phi_1||_{\infty}.
\end{align*}
Noting that
\begin{equation*}\label{times}
 \int_{A_j}\mathrm{d}t_j \cdots \mathrm{d}t_1= \frac{t^j}{j!},
\end{equation*}
we have
\begin{align*}
\left | \int_{\R^N}\int_{A_j}\widehat W_N(\textbf{V},0)\widehat
  P^*_{t_1;0}\phi_{j+1;j}\mathrm{d}m^{(N)}\mathrm{d}t_1 \dots
  \mathrm{d}t_j \right |
\leq (4t)^j \hspace{0.015in} ||\phi_1||_{\infty}.
\end{align*}
This is an estimate of a general term in the first series in the right
hand side of (\ref{seriesnew}). Hence, that series is uniformly
convergent in $N$ if $t<1/4$. For the second series in the right hand
side of (\ref{seriesnew}) using (\ref{estimateg}) and
(\ref{restestimate}), we first obtain
\begin{align*}
||\Gamma^{j-i}_{t_{j-(i-1)},\dots,t_1} \phi^R_{i+1;i}||_{\infty}\leq
4^{j-i} \frac{j!}{i!} ||\phi^R_{i+1;i}||_{\infty}\leq 4^{j-i}
\hspace{0.015in} \frac{j!}{i!}\hspace{0.015in} \frac{6
  \hspace{0.015in} i^2}{N-1}
\hspace{0.015in}||\phi_{i;i-1}||_{\infty}.
\end{align*}
Since $||\phi_{i:i-1}||\leq 4^{i-1} \hspace{0.015in} (i-1)! \hspace{0.015in} ||\phi_1||_{\infty}$, we have
\begin{align*}
||\Gamma^{j-i}_{t_{j-(i-1)},\dots,t_1} \phi^R_{i+1;i}||_{\infty}\leq
\frac{3i}{2(N-1)}\hspace{0.015in} 4^j\hspace{0.015in} j!
\hspace{0.015in} \hspace{0.015in} ||\phi_1||_{\infty}.
\end{align*}
Hence,

\begin{align*}
\left |\left|\sum_{i=1}^{j}\Gamma^{j-i}_{t_{j-(i-1)},\dots,t_1} \phi^R_{i+1;i} \right |\right|_{\infty}
\leq \frac{3}{2}\hspace{0.015in}\frac{4^j\hspace{0.015in}
  j!}{(N-1)}\left(\sum_{i=1}^{j} i\right) ||\phi_1||_{\infty}
\leq 3\frac{4^{j-1}\hspace{0.015in} j! \hspace{0.015in}(j+1)^2}{N-1}||\phi_1||_{\infty}.
\end{align*}
Finally, we arrive at
\begin{align}\label{esttermsseriers1}
& \left| \int_{\R^N}\int_{A_j}\widehat W_N(\textbf{V},0)\widehat
  P^*_{t_1,0}\left(\sum_{i=1}^{j}\Gamma^{j-i}_{t_{j-(i-1)},\dots,t_1}
    \phi^R_{i+1;i} \right)\mathrm{d}m^{(N)}\mathrm{d}t_1\dots
  \mathrm{d}t_j \right| \notag \\
& \hspace{0.3in}\leq \frac{3}{4}\hspace{0.015in} (4t)^j
\hspace{0.015in}\frac{(j+1)^2}{N-1}\hspace{0.015in}||\phi_1||_{\infty}.
\end{align}
We also get that the second series in the right hand side of
(\ref{seriesnew}) is uniformly convergent in $N$ if $t<1/4$.

Similarly to the computation above, for $\phi_{2;0}=\phi_2$ where
$\phi_2$ a function of the two variables $v_1,v_2$, and defining
inductively $\widehat P^*_{t_1,0}\phi_{j+2,j}$ by
(\ref{formulagplus}), again it follows that
\begin{align}\label{seriesnew1}
\nonumber &\sum_{j=1}^{\infty} \int_{\R^N}\int_{A_j}\widehat
W_N(\textbf{V},0)\widehat P^*_{t_1,0}\Gamma^j_{t,t_j,\dots,t_1}
\phi_{2;0} \mathrm{d}m^{(N)}\mathrm{d}t_1 \dots \mathrm{d}t_j \\
\nonumber &=\sum_{j=1}^{\infty} \int_{\R^N}\int_{A_j}\widehat
W_N(\textbf{V},0) \widehat
P^*_{t_1;0}\phi_{j+2;j}(v_1,\dots,v_{j+2};t,t_j,\dots,t_1)\mathrm{d}m^{(N)}\mathrm{d}t_1
\dots \mathrm{d}t_j\\
\nonumber &\hspace{0.3in}+\sum_{j=1}^{\infty}
\int_{\R^N}\int_{A_j}\widehat W_N(\textbf{V},0)\widehat
P^*_{t_1,0}\left(\sum_{i=1}^{j}
\Gamma^{j-i}_{t_{j-(i-1)},\dots,t_1} \phi^R_{i+2;i}(v_1,\dots,v_{i+2};t,\dots,t_{j-(i-1)}) \right)\\
&\hspace{2.3in}\mathrm{d}m^{(N)}\mathrm{d}t_1 \dots \mathrm{d}t_j.
\end{align}
By induction and (\ref{formulagplus}) we get
\begin{align*}
||\phi_{j+2;j}||_{\infty}\leq 4^j \hspace{0.015in} (j+1)! \hspace{0.015in} ||\phi_2||_{\infty},
\end{align*}
which together with (\ref{times}) yields
\begin{align*}
\left | \int_{\R^N}\int_{A_j}\widehat W_N(\textbf{V},0)\widehat
  P^*_{t_1,0}\phi_{j+2;j}\mathrm{d}m^{(N)}\mathrm{d}t_1 \dots
  \mathrm{d}t_j \right |
\leq (4t)^j \hspace{0.015in} (j+1) \hspace{0.015in} ||\phi_2||_{\infty}.
\end{align*}
This implies that the first series in the right hand side of
(\ref{seriesnew1}) is uniformly convergent in $N$ if $t<1/4$. Now,
using (\ref{estimateg}) and (\ref{restestimate}) we get
\begin{align*}
||\Gamma^{j-i}_{t_{j-(i-1)},\dots,t_1} \phi^R_{i+2;i}||_{\infty}\leq
4^{j-i} \frac{(j+1)!}{(i+1)!} ||\phi^R_{i+2;i}||_{\infty}
\leq 4^{j-i} \hspace{0.015in} \frac{(j+1)!}{(i+1)!} \hspace{0.015in}
\frac{6(i+1)^2}{N-1} \hspace{0.015in} ||\phi_{i+1;i-1}||_{\infty}.
\end{align*}
Using $||\phi_{i+1;i-1}||\leq 4^{i-1} \hspace{0.015in} i! \hspace{0.015in} ||\phi_2||_{\infty}$, we get
\begin{align*}
\left |\left |\Gamma^{j-i}_{t_{j-(i-1)},\dots,t_1} \phi^R_{i+2;i}\right |\right |_{\infty}
\leq \frac{3}{2}\hspace{0.015in} 4^j \hspace{0.015in}(j+1)!
\hspace{0.015in} \frac{(i+1)}{N-1} ||\phi_2||_{\infty},
\end{align*}
which implies that
\begin{align*}
||\sum_{i=1}^{j}\Gamma^{j-i}_{t_{j-(i-1)},\dots,t_1}
\phi^R_{i+2;i}||_{\infty}&\leq \frac{3}{2} \hspace{0.015in} 4^j
\hspace{0.015in} \frac{(j+1)!}{N-1} \left(\sum_{i=1}^{j} (i+1)\right)
||\phi_2||_{\infty}\\
&\leq \frac{3}{4(N-1)}\hspace{0.015in} 4^j \hspace{0.015in} j!
\hspace{0.015in} (j+1)^3 \hspace{0.015in} ||\phi_2||_{\infty}.
\end{align*}
Finally, using the last estimate, we have
\begin{align}\label{esttermsseriers2}
&\left| \int_{\R^N}\int_{A_j}\widehat W_N(\textbf{V},0)\widehat
  P^*_{t_1,0}\left(\sum_{i=1}^{j}\Gamma^{j-i}_{t_{j-(i-1)},\dots,t_1}
    \phi^R_{i+2;i} \right)\mathrm{d}m^{(N)}\mathrm{d}t_1 \dots
  \mathrm{d}t_j \right| \notag\\
&\hspace{0.3in}\leq \frac{3}{4}\hspace{0.015in} (4t)^j\hspace{0.015in}
\frac{(j+1)^3}{N-1}\hspace{0.015in} ||\phi_2||_{\infty}.
\end{align}
Therefore, the second series in (\ref{seriesnew1}) is also uniformly convergent in $N$ if $t<1/4$.
\\ \\
\underline{\textbf{STEP 5}}\\
In this last step we use the series representation
(\ref{differentfofmarginal}) to obtain that the second marginal of
$\widehat W_N(\textbf{V},t)$ can be written as product of two first
marginals of $\widehat W_N(\textbf{V},t)$ as $N$ tends to
infinity. From (\ref{differentfofmarginal}), (\ref{seriesnew}) and the
estimates in STEP 4 it follows for $0 \leq t <T$ where $T<1/4$ that
\begin{align}\label{firstmargt}
 \nonumber   &\int_{\R} \widehat
 f_1^N(v_1,t)\phi_{1;0}(v_1)\mathrm{d}v_1=\int_{\R} \widehat
 f_1^N(v_1,0)\widehat P^*_{t;0}\phi_{1;0}\mathrm{d}v_1+ \\
 \nonumber   &+\sum_{j=1}^{\infty}\int_{\R^{j+1}}\int_{A_j}\widehat
 f^{N}_{j+1}(v_1,\dots,v_{j+1},0)\\
    &\hspace{1in} \widehat
    P^*_{t_1,0}\phi_{j+1;j}(v_1,\dots,v_{j+1};t,t_j,\dots,t_1)\mathrm{d}m^{(j+1)}\mathrm{d}t_1
    \dots \mathrm{d}t_j+ R_{1,T,N}
\end{align}
where the series is absolutely convergent uniformly in $N$ and using (\ref{esttermsseriers1}), we have
\begin{equation}
||R_{1,T,N}||_{\infty}\leq \frac{3}{N-1} \left (\sum_{j=1}^{\infty}
  (4t)^j(j+1)^2 \right) ||\phi_1||_{\infty} \leq
\frac{C_1(T)}{N}||\phi_1||_{\infty}.
\end{equation}
The constant $C_1(T)$ depends only on $T$ and
$$C_1(T) \sim \frac{1}{(T-1/4)^3}.$$
Now, for a function
$\psi_{2;0}(v_1,v_2)=\phi_{1;0}(v_1)\varphi_{1;0}(v_2)$, where
$\psi_{j+2;j}$ is inductively defined by (\ref{formulagplus}) again it
follows for $0 \leq t<T$ that
\begin{align}
\label{eq:3.24}
 \nonumber   &\int_{\R}\int_{\R} \widehat
 f_2^N(v_1,v_2,t)\psi_2(v_1,v_2)\mathrm{d}v_1 \mathrm{d}v_2=
 \int_{\R}\int_{\R}\widehat f_2^N(v_1,v_2,0)\widehat
 P^*_{t;0}\psi_{2;0}\mathrm{d}v_1 \mathrm{d}v_2+ \\
 \nonumber   &+\sum_{j=1}^{\infty}\int_{\R^{j+2}}\int_{A_j}\widehat
 f^N_{j+2}(v_1,\dots,v_{j+2},0)\\
    &\hspace{1in}\widehat
    P^*_{t_1,0}\psi_{j+2;j}(v_1,\dots,v_{j+2};t,t_j,\dots,t_1)
    \mathrm{d}m^{(j+2)}\mathrm{d}t_1
    \dots \mathrm{d}t_j+ R_{2,T,N}
\end{align}
where the series is absolutely convergent uniformly in $N$. Using
(\ref{esttermsseriers2}), we get
\begin{equation}
||R_{2,T,N}||_{\infty}\leq \frac{3}{(N-1)} \left (\sum_{j=1}^{\infty}
  (4t)^j(j+1)^3 \right) ||\psi_2||_{\infty} \leq
\frac{C_2(T)}{N}||\psi_2||_{\infty},
\end{equation}
where $C_2(T)$ is a constant depending only on $T$, and
$$C_2(T) \sim \frac{1}{(T-1/4)^4}.$$
From the assumptions in the initial data (\ref{intialdataassumption}),
we now obtain
\begin{align}\label{eqintdata}
 \nonumber   &\int_{\R}\int_{\R} \widehat
 f_2^N(v_1,v_2,t)\psi_2(v_1,v_2)\mathrm{d}v_1 \mathrm{d}v_2=
 \int_{\R}\int_{\R}\widehat f_1^N(v_1,0)\widehat f_1^N(v_2,0)\widehat
 P^*_{t;0}\psi_{2;0}\mathrm{d}v_1 \mathrm{d}v_2+ \\
 \nonumber   &+\sum_{j=1}^{\infty}\int_{\R^{j+2}}\int_{A_j} \left (\prod_{i=1}^{j+2}\widehat f^N_{1}(v_i,0)\right)\\
 \nonumber   &\hspace{0.6in}\widehat{P}^*_{t_1,0}\psi_{j+2;j}(v_1,\dots,v_{j+2};t,t_j,\cdots,t_1)\mathrm{d}m^{(j+2)}\mathrm{d}t_1 \dots \mathrm{d}t_j\\
     &+ R_{2,T,N} + \widetilde R_{2,T,N},
\end{align}
where using (\ref{intialdataassumption}) and (\ref{estimateg}) yields
\begin{equation}
||\widetilde R_{2,T,N}||_{\infty}\le \frac{C_0}{N} \left
  (\sum_{j=0}^\infty(j+2)(j+1)(4t)^j \right)
||\psi_2||_{\infty}\leq \frac{C_3(T)}{N}||\psi_2||_{\infty}.
\end{equation}
Here $C_3(T)$ is constant depending only on $T$ and
$$C_3(T) \sim \frac{1}{(T-1/4)^3}.$$
The main contribution thus comes from the sum, and we want to show that this is equal to
\begin{align}\label{eqintdataprodA}
\nonumber &\bigg (\int_{\R}\widehat f_1^N(v_1,0)\widehat P^*_{t;0}\phi_{1;0}\mathrm{d}v_1+
 \sum_{k=1}^{\infty}\int_{\R^{k+1}}\int_{A_k} \bigg (\prod_{i=1}^{k+1}\widehat f^N_{1}(v_i,0)\bigg)\\
 \nonumber &\hspace{0.9in}\widehat P^*_{t_1,0}\phi_{k+1;k}(v_1,\dots,v_{k+1};t,t_k,\dots,t_1)
\mathrm{d}m^{(k+1)}\mathrm{d}t_1 \dots \mathrm{d}t_k \bigg)\times\\
\nonumber  &\times
  \bigg (\int_{\R}\widehat f_1^N(v_2,0)\widehat P^*_{t;0}\varphi_{1;0}\mathrm{d}v_2+
 \sum_{l=1}^{\infty}\int_{\R^{l+1}}\int_{A_l} \bigg (\prod_{i=1}^{l+1}\widehat f^N_{1}(v_i,0)\bigg)\\
 \nonumber &\qquad \widehat
 P^*_{t_1,0}\varphi_{l+1;l}(v_1,\dots,v_{l+1};t,s_l,\dots,s_1)
\mathrm{d}m^{(l+1)}\mathrm{d}s_1 \dots \mathrm{d}s_l \bigg)\\
\end{align}

Since $\psi_{2;0}(v_1,v_2)=\phi_{1;0}(v_1)\varphi_{1;0}(v_2)$, and the
operator $P^*_{t,s}$ acts independently on each velocity variable, we
have
$$\widehat P^*_{t,0}\psi_{2;0}(v_1,v_2)=\widehat P^*_{t,0}\phi_{1;0}(v_1)\hspace{0.015in}\widehat
P^*_{t,0}\varphi_{1;0}(v_2)$$

For the remaining terms, using (\ref{formulagplus}) the calculation
follows very much like in Kac's original work, but taking into account the times $t_i$
when new velocities are added to the original two.

Consider again the construction of the factors
$\widehat{P}^*_{t_1,0}\psi_{j+2;j}(v_1,\dots,v_{j+2};t,t_j,\cdots,t_1)$
in Eq.~\ref{eqintdata}.
These factors in turn consists of several
terms, where each term is constructed by adding new velocities as
described  in Step 2 and Step 3. But here the starting point consists
of two velocities, and the main contribution will come from terms
represented by two trees, rooted at $(v_1,t)$ and $(v_2,t)$
respectively. Already from Kac's original work it follows that adding all terms
in which the tree rooted at $v_1$ has $k+1$ leafs and the tree rooted
at $v_2$ has $l+1$ leafs would give  exactly those terms in the
product~(\ref{eqintdataprodA}) which come from multiplying the $k$-th
term in the first factor with the $l$-th term in the second factor
(using as always the symmetry with respect to permutation of the
variables), {\em if the time points were }  not important.
Consider the following two pairs of trees, representing terms where
three new velocities are added to the original two:

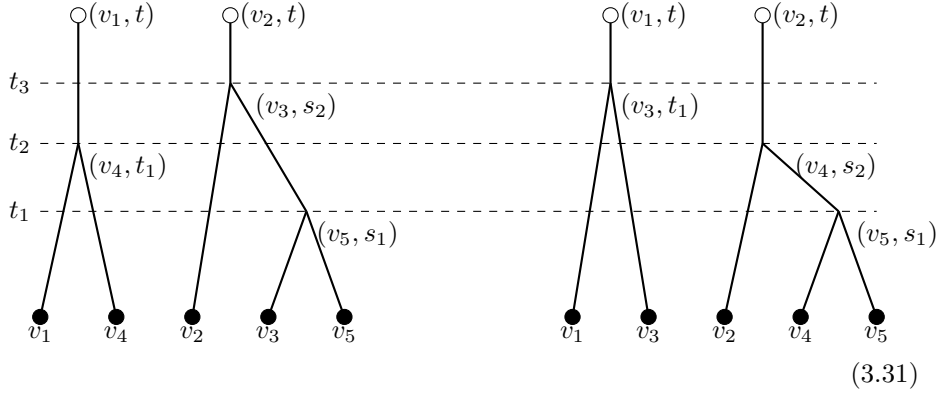
\begin{dispmathenv}
\label{eq:tree3}
\begin{tikzpicture}
{\color{black}
\draw[fill] (0,0) circle (0.1) node[below] {$v_1$};
\draw[fill] (1,0) circle (0.1) node[below] {$v_4$};
\draw[black] (0.5,4) circle (0.1) node[right] {$(v_1,t)$};
\draw[thick] (0.5,2.3) node[below right] {$(v_4,t_1)$} -- (0.5, 3.9) ;
\draw[thick] (0,0) -- (0.5,2.3);
\draw[thick] (1,0) -- (0.5,2.3);
\draw[fill] (2,0) circle (0.1)  node[below] {$v_2$};
\draw[fill] (3,0) circle (0.1)  node[below] {$v_3$};
\draw[fill] (4,0) circle (0.1)  node[below] {$v_5$};
\draw[black] (2.5,4) circle (0.1) node[right] {$(v_2,t)$};
\draw[thick] (2.5,3.1) node[below right] {$\;\;(v_3,s_2)$} -- (2.5, 3.9) ;
\draw[thick] (2,0) -- (2.5,3.1);
\draw[thick] (3.5,1.4) node[below right] {$(v_5,s_1)$} -- (2.5, 3.1) ;
\draw[thick] (3,0) -- (3.5,1.4);
\draw[thick] (4,0) -- (3.5,1.4);
\draw[fill] (7,0) circle (0.1) node[below] {$v_1$};
\draw[fill] (8,0) circle (0.1) node[below] {$v_3$};
\draw[black] (7.5,4) circle (0.1) node[right] {$(v_1,t)$};
\draw[thick] (7.5,3.1) node[below right] {$(v_3,t_1)$} -- (7.5, 3.9) ;
\draw[thick] (7,0) -- (7.5,3.1);
\draw[thick] (8,0) -- (7.5,3.1);
\draw[fill] (9,0) circle (0.1)  node[below] {$v_2$};
\draw[fill] (10,0) circle (0.1)  node[below] {$v_4$};
\draw[fill] (11,0) circle (0.1)  node[below] {$v_5$};
\draw[black] (9.5,4) circle (0.1) node[right] {$(v_2,t)$};
\draw[thick] (9.5,2.3) node[below right] {$\;\;\;(v_4,s_2)$} -- (9.5, 3.9) ;
\draw[thick] (9,0) -- (9.5,2.3);
\draw[thick] (10.5,1.4) node[below right] {$\;(v_5,s_1)$} -- (9.5, 2.3) ;
\draw[thick] (10,0) -- (10.5,1.4);
\draw[thick] (11,0) -- (10.5,1.4);
\draw[dashed] (0,1.4) node[left] {$t_1$} -- (11,1.4) ;
\draw[dashed] (0,2.3) node[left] {$t_2$} -- (11,2.3) ;
\draw[dashed] (0,3.1) node[left] {$t_3$} -- (11,3.1) ;
}
\end{tikzpicture}
\end{dispmathenv}

In the trees, the time  points of added velocities are denoted $s_j$
for the  tree rooted at $v_1$ and $t_j$ for the tree rooted at $v_2$,
as if they were representing terms in the
product~(\ref{eqintdataprodA}), and to the left the same time points are indexed
by only $t_j$-s, as when representing a term in~(\ref{eqintdata}).
The examples to the left and right would be identical in Kac's original
model, but here they are different, and we need a small computation to
see that after carrying out the integrals, we do get the correct result.

Consider an arbitrary function $u\in C( \R^j)$. Then
\begin{align*}
  \int_{A_j} u(t_j,...,t_1)\mathrm{d}t_1 \mathrm{d}t_2 \dots \mathrm{d}t_j &=
    \frac{t^j}{j!} \E\left[u(\tau_j,...,\tau_1)\right]\,,
\end{align*}
where $(\tau_i)_{i=1}^j$ is the increasing reordering of $i$
independent random variables uniformly distributed on $[0,t]$. In the
same way, for $u\in C(\R^{k+l})$,
\begin{align*}
  \int_{A_k\times A_{l}} u(t_k,...,t_1,s_{l},...,s_1)
  \mathrm{d}t_1\dots \mathrm{d}t_k \mathrm{d}s_1\dots \mathrm{d}s_{l} &=
  \frac{t^{k+l}}{k!\,l!}\E\left[u(\tau_k,...,\tau_1,\sigma_l,...,\sigma_1)\right]\,,
\end{align*}
where $(\tau_i)_{i=1}^k$ and $(\sigma_i)_{i=1}^{l}$  are two
increasing lists of time points obtained as reorderings of i.i.d random variables
as above. But these independent increasing lists can also be obtained
by taking $k+l$ independent random variables, uniformly distributed
on $[0,t]$, reordering them in increasing order, and then making a
random choice of $k$ of them to form $(\tau_i)_{i=1}^k$, leaving the
remaining ones for $(\sigma_i)_{i=1}^{l}$.
Hence
\begin{align*}
   \frac{t^{k+l}}{k!\,l!}&\E\left[u(\tau_k,...,\tau_1,\sigma_l,...,\sigma_1)\right]
   =\\
  &\frac{t^{k+l}}{k!\,l!}
 \left( \frac{k! l!}{(k+l)!}\sum  \frac{1}{\left|A_{k+l}\right|}\int_{A_{k+l}}
      u(t_{L_k},..., t_{L_1}, t_{R_{l}},...,t_{R_1})
      \mathrm{d}t_1\dots \mathrm{d}t_{k+l} \right)
\end{align*}
where the sum is taken over all partitions of $t_1,...,t_{l+k}$
into to increasing sequences $t_{L_1},...,t_{L_k}$ and $t_{R_1},...,t_{R_{l}}$.
Because $\left|A_{k+l}\right| = \frac{t^{l+k}}{(k+l)!}$ we
see that
\begin{align*}
   \int_{A_k\times A_{l}} u(t_k,...,t_1,s_{l},...,s_1)
  & \mathrm{d}t_1\dots \mathrm{d}t_k \mathrm{d}s_1\dots \mathrm{d}s_{l} = \\
  & \sum  \int_{A_{k+l}}
      u(t_{L_k},..., t_{L_1}, t_{R_{l}},...,t_{R_1}) \,
      \mathrm{d}t_1\dots \mathrm{d}t_{k+l}
\end{align*}
We may now conclude by taking
\begin{align*}
  u(t_{k+l},\dots,t_1) =
    \widehat{P}^*_{\min(t_1,s_1),0}\bigg(&\widehat{P}^*_{t_1,\min(t_1,s_1)}
      \phi_{k+1;k}(v_1,v_3,\dots,v_{k+2};t,t_k,\dots,t_1)
      \\
      &\widehat{P}^*_{s_1,\min(t_1,s_1)}\varphi_{l+1;l}(v_2,v_{k+3},\dots,v_{k+l+2};t,t_l,\dots,t_1)
\bigg)\,.
\end{align*}
Abbreviating $R_{T,N}=R_{2,T,N}+ \widetilde R_{2,T,N}$ and $C(T)=C_2(T)+C_3(T)$,
we have
$$||R_{T,N}||_{\infty} \leq \frac{C(T)}{N}||\psi_2||_{\infty}.$$
Thus, for $0\leq t \leq T$ we have that
\begin{align*}
&\int_{\R}\int_{\R} \widehat f_2^N(v_1,v_2,t)\phi(v_1)\varphi(v_2)\mathrm{d}v_1 \mathrm{d}v_2\\
&\hspace{0.4in}=\int_{\R} \widehat f_1^N(v_1,t)\phi(v_1)\mathrm{d}v_1 \int_{\R} \widehat f_1^N(v_2,t)\varphi(v_2)\mathrm{d}v_2
+R_{T,N},\\
\end{align*}
where $R_{T,N}\rightarrow 0$ as $1/N$. Since the $T$ is independent of the initial distribution, we can take $t_1$ with $0<t_1<T$ and repeat the proof to extend the result to $t_1\leq t <t_1+T$. Clearly, the constant $C(T)$ also changes, but it still depends on time and the factor $1/N$ remains unchanged.  We can continue in this way to cover any time range $0\leq t \leq T<\infty$.  This concludes the proof.\\
\end{proof}

Combining Theorem \ref{thm_quant_thermostat} and Lemma \ref{momentconserv} yields the following corollary which will be needed later.
\begin{corollary}\label{corthmquench}
Assume that $\widehat m_{6,N}(0)<\infty$. Let $\psi$ and $\phi$ be two functions such that
$$\psi(v_1)\leq C_1(1+v_1^2) \hspace{0.2in}\text{and} \hspace{0.2in} \phi(v_2)\leq C_1(1+v_2^2),$$
where $C_1$ and $C_2$ are two positive constants. Then we have
\begin{align*}
&\int_{\R}\int_{\R} \widehat f_2^N(v_1,v_2,t)\psi(v_1)\phi(v_2)\mathrm{d}v_1 \mathrm{d}v_2= \\
&\hspace{0.5in}\int_{\R} \widehat f_1^N(v_1,t)\psi(v_1) \mathrm{d}v_1\int_{\R} \widehat f_1^N(v_2,t)\phi(v_2) \mathrm{d}v_2+ S_{T,\psi,\phi}+\widetilde{S}_{T,\psi,\phi},
\end{align*}
where
$$|S_{T,\psi,\phi}|\leq \frac{C(T)}{\sqrt{N}} \hspace{0.2in}\text{and} \hspace{0.2in} |\widetilde{S}_{T,\psi,\phi}| \leq \frac{\widetilde{C}}{\sqrt{N}}\widehat{m}_{6,N}(T).$$
Here $\widetilde{C}$ is a positive constant and $C(T)$ is given by Theorem \ref{thm_quant_thermostat}.
\end{corollary}

\begin{proof}
Let $0<\alpha<1$, we have
\begin{align*}
&\int_{\R}\int_{\R} \widehat f_2^N(v_1,v_2,t)\psi(v_1)\phi(v_2)\mathrm{d}v_1 \mathrm{d}v_2=\\
&\hspace{0.2in}\int_{\R}\int_{\R} \widehat f_2^N(v_1,v_2,t)\psi(v_1)\phi(v_2)\mathbbm{1}_{\{|v_1|\leq N^\alpha\}}\mathbbm{1}_{\{|v_2|\leq N^\alpha\}}\mathrm{d}v_1 \mathrm{d}v_2\\
&\hspace{0.4in}+ \int_{\R}\int_{\R} \widehat f_2^N(v_1,v_2,t)\psi(v_1)\phi(v_2)(1-\mathbbm{1}_{\{|v_1|\leq N^\alpha\}}\mathbbm{1}_{\{|v_2|\leq N^\alpha\}})\mathrm{d}v_1 \mathrm{d}v_2 := I+\widetilde{S}_{T,\psi ,\phi}.
\end{align*}
Choosing smooth cutoff functions to approximate the characteristic functions in $I$ it follows by using Theorem \ref{thm_quant_thermostat}
\begin{align}
I=\int_{\R} \widehat f_1^N(v_1,t)\psi(v_1) \mathrm{d}v_1\int_{\R} \widehat f_1^N(v_2,t)\phi(v_2) \mathrm{d}v_2+ S_{T,\psi,\phi}
\end{align}
where
\begin{equation}
|S_{T,\psi,\phi}|\leq \frac{C(T)}{N^{1-2\alpha}}.
\end{equation}
In $\widetilde{S}_{T,\psi ,\phi}$, either $|v_1|$ or $|v_2|$ is larger
than $N^\alpha$, hence $|v_1|^2+|v_2|^2 \geq N^{2\alpha}$. Hence,
using the inequalities $2ab<a^2+b^2$ and\footnote{this is a consequence of the first inequality} $a^2 b + a b^2
\le |a|^3+|b|^3$, where $a,b\in\R$
\begin{align*}
|\widetilde{S}_{T,\psi,\phi}|&\leq \frac{\widetilde{C}}{N^{2\alpha}} \int_{\R}\int_{\R} \widehat f_2^N(v_1,v_2,t)(1+v_1^2)(1+v_2^2)|v_1|^2 \mathrm{d}v_1 \mathrm{d}v_2\\
&\leq \frac{\widetilde{C}}{N^{2\alpha}} \int_{\R}\int_{\R} \widehat f_2^N(v_1,v_2,t)\big((1+v_1^6)+(1+v_1^6)\big)\mathrm{d}v_1 \mathrm{d}v_2\\
&\leq \frac{\widetilde{C}}{N^{2\alpha}}\widehat{m}_{6,N}(T),
\end{align*}
where $\widetilde{C}$ is a generic constant.
The proof may then be concluded by choosing $\alpha=1/4$.
\end{proof}


\section{Pathwise comparison of the processes}
\noindent We now have two master equations, the master equation (\ref{Nparticlemaster}) and the quenched master equation (\ref{quenchmaster}) where the latter propagates chaos according to the Kac's definition. Following~\cite{bcelm} we now consider the two stochastic processes $$\textbf{V}(t)=~(v_1(t),\dots,v_N(t))$$
corresponding to the master equation (\ref{Nparticlemaster}) and
$$\widehat{\textbf{V}}(t)=(\widehat{v}_1(t),\dots,\widehat{v}_N(t))$$
corresponding to the quenched master equation (\ref{quenchmaster}). The aim of this section is to compare these two processes and show that when $N$ is large, with high probability the paths of the two processes are close to each other. The starting point is to find a formula for the difference between the paths of the stochastic processes. There are two sources of randomness in these processes, the first coming from initial data while the second is from the collision history, i.e., the collision times $t_k$ and the pair of velocities $(v_i(t_k),v_j(t_k))$ or  $(\widehat{v}_i(t_k),\widehat{v}_j(t_k))$ participating in this collision process and the random collision parameter $\theta_k$. Let $\textbf{V}_0$ be the vector of initial velocities and
$\omega$ the collision history. We assume that the two stochastic processes have the same initial velocities and collision history. For each process there is a unique sample path given $\textbf{V}_0$ and $\omega$.
Let $\textbf{V}(t,\textbf{V}_0,\omega)$ and $\widehat{\textbf{V}}(t,\textbf{V}_0,\omega)$ denote these sample paths.
As in~\cite{bcelm}, we define
\begin{equation}
\Psi_t(\textbf{V}) \hspace{0.5in}\text{ and }\hspace{0.5in} \widehat \Psi_{s,t}(\textbf{V})
\end{equation}
to be the flows generated by the autonomous dynamics
(\ref{Nparticled}) and non-autonomous dynamics (\ref{dynamic_quench}),
respectively. Given a collision history $\omega$, consider the time
interval $[s,t]$ where no collision occur at time $s$ and $t$, and
suppose that there are $n$ collisions in this time interval with
collision times $t_k$, i.e., $s<t_1<t_2<\cdots
<t_{n-1}<t_n<t$. Moreover, denote by $\{(t_k,(i,j),\theta_k)\}$ the
collision history in the time interval $[s,t]$ , where $t_k$ is the
time for the $k$-th collision, $(i,j)$ the indices of the colliding particles and $\theta_k$ is the random collision
parameter. Then the stochastic process corresponding to the master
equation (\ref{Nparticlemaster}), starting from $\textbf{V}_s$ at time
$s$, has the path
\begin{equation}\label{pathstochNp}
\textbf{V}(t,\textbf{V}_s,\omega):=M_{t,s}(\textbf{V}_s,\omega)=\Psi_{t-t_n}
\circ R_{{ij}}(\theta_n) \circ \cdots \circ R_{{ij}}(\theta_1)\circ
\Psi_{t_1-s}(\textbf{V}_s)
\end{equation}
while the stochastic process corresponding the quenched master
equation (\ref{quenchmaster}) has the path
\begin{equation}
\widehat{\textbf{V}}(t,\textbf{V}_s,\omega):=\widehat M_{t,s}(\textbf{V}_s,\omega)=\widehat \Psi_{t_n,t} \circ R_{{ij}}(\theta_n) \circ \cdots \circ R_{{ij}}(\theta_1)\circ \widehat \Psi_{s,t_1}(\textbf{V}_s).
\end{equation}
In what follows we shall use the following two norms: Given a vector $\textbf{V}=(v_1,\dots,v_N) \in\R^N$,
\begin{equation}
||\textbf{V}||=\sqrt{\sum_{i=1}^{N}v_i^2} \hspace{0.3in}\text{and} \hspace{0.3in}
||\textbf{V}||_N=\frac{||\textbf{V}||}{\sqrt{N}} = \sqrt{\frac{1}{N}\sum_{i=1}^{N}v_i^2}.
\end{equation}
The goal of this section is to estimate
\begin{equation}\label{diffpath}
\prob \{||\textbf{V}(t,\textbf{V}_0,\omega)- \widehat{\textbf{V}}(t,\textbf{V}_0,\omega)||_N \geq \epsilon\},
\end{equation}
where $\epsilon$ is given positive number. In order to do that we
first need to find an expression for
$\textbf{V}(t,\textbf{V}_0,\omega)-
\widehat{\textbf{V}}(t,\textbf{V}_0,\omega)$. For completeness we
carefully explain the steps. The first step is to find an expression
for the difference of the paths of the two processes between
collisions, i.e., the difference between the flows
$\Psi_t(\textbf{V})$ and $\widehat \Psi_{0,t}(\textbf{V})$. We recall
the following useful formula for the difference between the product of
a sequence of real numbers. If $(a_1,a_2,\dots,a_n)\in \R^N$ and
$(b_1,b_2,\dots,b_n)\in \R^N$ then
\begin{equation}\label{identity_prod}
\prod_{i=1}^{N} a_i - \prod_{j=1}^{N}b_j=\sum_{j}\prod_{i<j}a_i(a_j-b_j)\prod_{i>j}b_i.
\end{equation}

\begin{lemma}
Between time $s$ and time $t$, the difference between
$\Psi_t(\textbf{V})$ and  $\widehat \Psi_{s,t}(\textbf{V})$ is given
by
\begin{equation}\label{diff_flow}
\Psi_{t-s}(\textbf{V})-\widehat \Psi_{s,t}(\textbf{V})=
\int_s^t D\Psi_{t-\tau}(\widehat
\Psi_{s,\tau}(\textbf{V}))(\textbf{F}(\widehat
\Psi_{s,\tau}(\textbf{V}))-\widehat{\textbf{F}}(\widehat
\Psi_{s,\tau}(\textbf{V})))\mathrm{d}\tau
\end{equation}
where $D\Psi_t(\textbf{V})$  is the differential of the flow starting
at \textbf{V} at time $s$.
\end{lemma}
\begin{proof}
The flows can be written as
\begin{equation}
\Psi_{t-s}(\textbf{V})=\Psi_{t-t_n}\circ \Psi_{t_n-t_{n-1}} \circ \cdots
\circ \Psi_{t_1-s}(\textbf{V})
\end{equation}
and
\begin{equation}
\widehat \Psi_{s,t}(\textbf{V})=\widehat \Psi_{t_n,t}\circ \widehat
\Psi_{t_{n-1},t_n}\circ  \cdots \circ \widehat
\Psi_{s,t_1}(\textbf{V}),
\end{equation}
where $s=t_0 < t_1 < t_2 < \cdots < t_n <t_{n+1}=t$. In the following
expressions the symbol $\prod$ is used to denote composition.
Using the identity (\ref{identity_prod}), we have
\begin{eqnarray*}
\Psi_t(\textbf{V})-\widehat \Psi_{0,t}(\textbf{V})&=&
\sum_j \left(  \prod_{i<j}\Psi_{t_{n+2-i}-t_{n+1-i}}\right)
\circ\left(\Psi_{t_{n+2-j}-t_{n+1-j}}-\widehat
  \Psi_{t_{n+1-j},t_{n+2-j}} \right)\\
&&\hspace{1.6 in} \circ \left( \prod_{i>j}\widehat
  \Psi_{t_{n+1-i},t_{n+2-i}}\right)\\
&=&\sum_j \Psi_{t-t_{n+2-j}} \circ
\left(\Psi_{t_{n+2-j}-t_{n+1-j}}-\widehat \Psi_{t_{n+1-j},t_{n+2-j}}
\right) \circ \widehat \Psi_{s,t_{n+1-j}}\\
&=&\sum_j \Psi_{t-t_{n+2-j}} \circ
\left(\Psi_{t_{n+2-j}-t_{n+1-j}}-Id + Id - \widehat \Psi_{t_{n+1-j},t_{n+2-j}}
\right) \circ \widehat \Psi_{s,t_{n+1-j}}\\
&=&\sum_j \Psi_{t-t_{n+1-j}} \circ \widehat
\Psi_{s,t_{n+1-j}}-\Psi_{t-t_{n+2-j}} \circ \widehat
\Psi_{s,t_{n+1-j}}\\
&&+\sum_j \Psi_{t-t_{n+2-j}} \circ \left(\widehat
\Psi_{s,t_{n+1-j}}-\widehat \Psi_{s,t_{n+2-j}}\right)\\
&=:&I_1^n+I_2^n.
\end{eqnarray*}
where $Id$ denotes the identity operator. Let $\triangle
t=t_{n+2-j}-t_{n+1-j}$.
By the flow property and a first order Taylor expansion
($\textbf{F}(\cdot)$ is differentiable), we have
\begin{align*}
\nonumber \Psi_{t-t_{n+2-j}}(\widehat
\Psi_{s,t_{n+1-j}})&=\Psi_{t-t_{n+1-j}}(\Psi_{-\triangle t}(\widehat
\Psi_{s,t_{n+1-j}}))\\
&= \Psi_{t-t_{n+1-j}}(\widehat \Psi_{s,t_{n+1-j}}- F(\widehat
\Psi_{s,t_{n+1-j}})\triangle t) + \Ordo((\triangle t)^2).
\end{align*}
Making a first order Taylor expansion in the last equality around
$\widehat \Psi_{s,t_{n+1-j}}$ yields
\begin{align}\label{flowI1}
\nonumber &\Psi_{t-t_{n+2-j}}(\widehat \Psi_{s,t_{n+1-j}})\\
&\hspace{0.3in}=\Psi_{t-t_{n+1-j}}(\widehat \Psi_{s,t_{n+1-j}})- D
\Psi_{t-t_{n+1-j}}(\widehat \Psi_{s,t_{n+1-j}}) F(\widehat
\Psi_{s,t_{n+1-j}})\triangle t + \Ordo((\triangle t)^2).
\end{align}
Plugging (\ref{flowI1}) into $I_1^n$, we finally have
\begin{equation*}
\lim_{n\rightarrow \infty} I_1^n=\int_s^t D\Psi_{t-\tau}(\widehat
\Psi_{s,\tau}(\textbf{V}))(\textbf{F}(\widehat
\Psi_{s,\tau}(\textbf{V})))\mathrm{d}\tau.
\end{equation*}
Next, by a first order Taylor expansion
\begin{align*}
\widehat \Psi_{s,t_{n+2-j}}-\widehat
\Psi_{s,t_{n+1-j}}=\widehat{\textbf{F}}(\widehat \Psi_{s,t_{n+1-j}})
\triangle t + \Ordo((\triangle t)^2),
\end{align*}
which together with another first order Taylor expansion leads to
\begin{align}\label{flowI2}
 \nonumber &\Psi_{t-t_{n+2-j}}(\widehat \Psi_{s,t_{n+2-j}})\\
 &=\Psi_{t-t_{n+1-j}}(\widehat \Psi_{s,t_{n+1-j}})+D
 \Psi_{t-t_{n+1-j}}(\widehat
 \Psi_{s,t_{n+1-j}})\widehat{\textbf{F}}(\widehat \Psi_{s,t_{n+1-j}})
 \triangle t + \Ordo((\triangle t)^2).
\end{align}
Plugging (\ref{flowI2}) into $I_2^n$, we obtain
\begin{equation*}
\lim_{n\rightarrow \infty} I_2^n=-\int_s^t D\Psi_{t-\tau}(\widehat \Psi_{s,\tau}(\textbf{V}))(\widehat{\textbf{F}}(\widehat \Psi_{s,\tau}(\textbf{V})))\mathrm{d}\tau.
\end{equation*}
This completes the proof.
\end{proof}

Since we assumed that our  two stochastic processes $\textbf{V}(t,\textbf{V}_0,\omega)$ and $\widehat{\textbf{V}}(t,\textbf{V}_0,\omega)$ have the same collision history and $R_{ij}(\theta_k)$ is a norm preserving linear operator , i.e.,
$$||R_{ij}(\theta_k)\textbf{V}||_N=||\textbf{V}||_N,$$
we can  extend (\ref{diff_flow}) to include the collisions and hence obtain a formula for the difference of the path of the two processes:
\begin{eqnarray}\label{pathdiff}
\nonumber \textbf{V}(t,\textbf{V}_0,\omega)- \widehat{\textbf{V}}(t,\textbf{V}_0,\omega)&=&
\int_0^t D M_{t,s}(\widehat{\textbf{V}}(s,\textbf{V}_0,\omega))\\
&& \times[\textbf{F}(\widehat{\textbf{V}}(s,\textbf{V}_0,\omega))-\widehat{\textbf{F}}(\widehat{\textbf{V}}(s,\textbf{V}_0,\omega))]\mathrm{d}s.
\end{eqnarray}
Having this formula, we see that in order to estimate
(\ref{diffpath}), the quantities
\begin{equation}\label{flow_estimate}
|||D M_{t,s}(\widehat{\textbf{V}}(s,\textbf{V}_0,\omega))|||
:=\sup_{||\textbf{V}||_N=1}||D
M_{t,s}(\widehat{\textbf{V}}(s,\textbf{V}_0,\omega))\textbf{V}||_N
\end{equation}
and
\begin{equation}\label{force_estimate}
\int_0^t
||\textbf{F}(\widehat{\textbf{V}}(s,\textbf{V}_0,\omega))
-\widehat{\textbf{F}}
(\widehat{\textbf{V}}(s,\textbf{V}_0,\omega))||_N\mathrm{d}s
\end{equation}
need to be estimated. We start with (\ref{force_estimate}). From
Theorem \ref{thm_quant_thermostat} we know that the quenched master
equation propagates chaos and we also know the rate of
convergence. This implies that for large $N$,
$J(\widehat{\textbf{V}}(t,\textbf{V}_0,\omega))$ should be close to
$\widehat{J}_N(t)$. More precisely, we have the following proposition
which corresponds to proposition $3.3$ in~\cite{bcelm} but the
difference appears in that their quenched master equation propagates
independence, while here we only have propagation of chaos:

\begin{proposition}\label{propbddforce}
Let $\widehat W_N(\textbf{V},0)$ be a probability density on $\R^N$
that satisfies the assumptions in Theorem
\ref{thm_quant_thermostat}. Suppose also that
$$ \widehat U_N>0 \hspace{0.3in}\mbox{and}\hspace{0.3in} \widehat
m_{6,N}(0)<\infty .$$
Then, for $t<T$
\begin{align*}
\E\bigg(\int_0^t||\textbf{F}(\widehat{\textbf{V}}
&(s,\textbf{V}_0,\omega))-\widehat{\textbf{F}}
(\widehat{\textbf{V}}(s,\textbf{V}_0,\omega))
||_N\mathrm{d}s \bigg)
\leq
 \frac{E T}{\widehat U_N N^{1/4}} \left(\sqrt{A_1(T)}+ \sqrt{A_2(T)} \right),
\end{align*}
where
\begin{align*}
A_1(T)&=\widehat U_N (\widehat U_N+C(T)+\widetilde{C}\widehat{m}_{6,N}(T)),\\
A_2(T)&= \widehat{m}_{6,N}(T)^{2/3}+\widehat U_N^2+C(T)+\widetilde{C}\widehat{m}_{6,N}(T).
\end{align*}
Here $\widetilde{C}$ is positive constant and  $C(T)$ is given by Theorem \ref{thm_quant_thermostat}.
\end{proposition}
\begin{proof}
The difference componentwise of the forces $\textbf{F}(\widehat{\textbf{V}}(s,\textbf{V}_0,\omega))$, $\widehat{\textbf{F}}(\widehat{\textbf{V}}(s,\textbf{V}_0,\omega))$ at time $s$ given by
\begin{equation*}
(\textbf{F}(\widehat{\textbf{V}}(s,\textbf{V}_0,\omega))-\widehat{\textbf{F}}(\widehat{\textbf{V}}(s,\textbf{V}_0,\omega)))_i
=E \left(\frac{J(\widehat{\textbf{V}})}{U(\widehat{\textbf{V}})}-\frac{\widehat{J}_N(s)}{\widehat U_N}\right)\hat{v}_i(s,\textbf{V}_0,\omega).
\end{equation*}
Moreover, we can write
\begin{equation*}
\left(\frac{J(\widehat{\textbf{V}})}{U(\widehat{\textbf{V}})} -\frac{\widehat{J}_N(s)}{\widehat U_N}\right)=\frac{J(\widehat{\textbf{V}})-\widehat{J}_N(s)}{\widehat U_N}
+\left(\frac{1}{U(\widehat{\textbf{V}})}-\frac{1}{\widehat U_N} \right)J(\widehat{\textbf{V}}),
\end{equation*}
and
\begin{equation*}
\left(\frac{1}{U(\widehat{\textbf{V}})}-\frac{1}{\widehat U_N} \right)J(\widehat{\textbf{V}})=(\widehat U_N-U(\widehat{\textbf{V}}))\frac{J(\widehat{\textbf{V}})}{U(\widehat{\textbf{V}})\widehat U_N}.
\end{equation*}
Using the inequality $|J(\widehat{\textbf{V}})|\leq \sqrt{U(\widehat{\textbf{V}})}$ together with the triangle inequality we arrive at
\begin{align*}
||\textbf{F}(\widehat{\textbf{V}}&(s,\textbf{V}_0,\omega))-\widehat{\textbf{F}}(\widehat{\textbf{V}}(s,\textbf{V}_0,\omega))||_N  \leq  \\
&\frac{E}{\widehat U_N} \bigg |J(\widehat{\textbf{V}}(s,\textbf{V}_0,\omega))-\widehat{J}_N(s)\bigg | \sqrt{U(\widehat{\textbf{V}}(s,\textbf{V}_0,\omega))}+
\frac{E}{\widehat U_N}\bigg |U(\widehat{\textbf{V}}(s,\textbf{V}_0,\omega))-\widehat U_N \bigg |.
\end{align*}
Integrating both sides of the last inequality over the interval $[0,t]$, taking expectation with respect to $\widehat W_N(\textbf{V},t)$ and using the Cauchy Schwartz inequality leads to
\begin{align}\label{ExpFhat}
\nonumber \E \bigg(\int_0^t & ||\textbf{F}(\widehat{\textbf{V}}(s,\textbf{V}_0,\omega))-\widehat{\textbf{F}}(\widehat{\textbf{V}}(s,\textbf{V}_0,\omega))||_N \mathrm{d}s \bigg) \leq \\
\nonumber &\frac{E}{\widehat U_N}\int_0^t \big(\E \big |J(\widehat{\textbf{V}}(s,\textbf{V}_0,\omega))-\widehat{J}_N(s) \big|^2 \big)^{1/2}\big(\E(U(\widehat{\textbf{V}}(s,\textbf{V}_0,\omega)))\big)^{1/2}\mathrm{d}s\\
&\hspace{0.2in}+\frac{E}{\widehat U_N}\int_0^t \big(\E \big|U(\widehat{\textbf{V}}(s,\textbf{V}_0,\omega))-\widehat U_N \big|^2\big)^{1/2}\mathrm{d}s.
\end{align}
First, by definition it follows
\begin{equation*}
\E(U(\widehat{\textbf{V}}(s,\textbf{V}_0,\omega)))=\widehat U_N.
\end{equation*}
To estimate $\E \big |J(\widehat{\textbf{V}}(s,\textbf{V}_0,\omega))-\widehat{J}_N(s)\big |^2$, we first note that
\begin{equation*}
J(\widehat{\textbf{V}}(s,\textbf{V}_0,\omega))-\widehat{J}_N(s)= \frac{1}{N}\sum_{j=1}^{N}(\widehat{v}_j(s)-\widehat{J}_N(s))
\end{equation*}
where $\widehat{v}_j(s)=\widehat{v}_j(s,\textbf{V}_0,\omega)$. From this it follows
\begin{align*}
&(J(\widehat{\textbf{V}}(s,\textbf{V}_0,\omega))-\widehat{J}_N(s) )^2\\
&\hspace{0.2in}= \frac{1}{N^2}
\left[\sum_{j=1}^{N}(\hat{v}_j(s)-\widehat{J}_N(s))^2+2\sum_{i=1}^{N}\sum_{j>i}^{N}(\hat{v}_i(s)-\widehat{J}_N(s))(\hat{v}_j(s)-\widehat{J}_N(s))\right].
\end{align*}
Using Corollary \ref{corthmquench} with $\psi(v)=(v-\widehat{J}_N(s))$ and $\phi(w)=(w-\widehat{J}_N(s))$, we get
\begin{align*}
\E |J(\widehat{\textbf{V}}&(s,\textbf{V}_0,\omega))-\widehat{J}_N(s)|^2 = \frac{1}{N}\int_\R (v-\widehat{J}_N(s))^2 \widehat{f}_1^N(v,s)\mathrm{d}v\\
&+\frac{N-1}{N}\int_\R\int_\R (v-\widehat{J}_N(s))(w-\widehat{J}_N(s))\widehat{f}_1^N(v,s)\widehat{f}_1^N(w,s)\mathrm{d}v \mathrm{d}w
+S_{T,\psi,\phi}+\widetilde S_{T,\psi,\phi},
\end{align*}
where
$$|S_{T,\psi,\phi}|\leq \frac{C(T)}{\sqrt{N}} \hspace{0.2in}\text{and} \hspace{0.2in} |\widetilde{S}_{T,\psi,\phi}| \leq \frac{\widetilde{C}}{\sqrt{N}}\widehat{m}_{6,N}(T),$$
with $\widetilde{C}$ being a positive constant and $C(T)$ is given by Theorem \ref{thm_quant_thermostat}.
Estimating the first of the two last integrals yields
\begin{align*}
\int_\R (v-\widehat{J}_N(s))^2 \widehat{f}_1^N(v,s) \mathrm{d}v&=\int_\R v^2 \widehat{f}_1^N(v,s)\mathrm{d}v-2\widehat{J}_N(s)\int_\R v\widehat{f}_1^N( v,s)\mathrm{d}v+\widehat{J}_N(s)^2\\
&=\widehat U_N -\widehat{J}_N(s)^2 \leq \widehat U_N,
\end{align*}
and by symmetry
\begin{align*}
\int_\R\int_\R (v-\widehat{J}_N(s))(w-\widehat{J}_N(s))\widehat{f}_1^N(v,s)\widehat{f}_1^N(w,s)\mathrm{d}v \mathrm{d}w=0.
\end{align*}
Combining these inequalities, we have
\begin{equation}
\E \big |J(\widehat{\textbf{V}}(s,\textbf{V}_0,\omega))-\widehat{J}_N(s)\big |^2 \widehat U_N \leq
\frac{\widehat U_N (\widehat U_N+C(T)+\widetilde{C}\widehat{m}_{6,N}(T))}{\sqrt{N}}:=
\frac{A_1(T)}{\sqrt{N}}.
\end{equation}
A similar computation like the one we preformed to estimate $\E |J(\widehat{\textbf{V}}(s,\textbf{V}_0,\omega))-\widehat{J}_N(s)|^2$ also yields
\begin{equation}
\E \big |U(\widehat{\textbf{V}}(s,\textbf{V}_0,\omega))-\widehat U_N \big |^2
\leq \frac{\widehat{m}_{6,N}(T)^{2/3}+\widehat U_N^2+C(T)+\widetilde{C}\widehat{m}_{6,N}(T)}{\sqrt{N}}:=\frac{A_2(T)}{\sqrt{N}}
\end{equation}
Collecting all the inequalities above and plugging them into (\ref{ExpFhat}), we finally get
\begin{align*}
\E \bigg(\int_0^t & ||\textbf{F}(\widehat{\textbf{V}}(s,\textbf{V}_0,\omega))-\widehat{\textbf{F}}(\widehat{\textbf{V}}(s,\textbf{V}_0,\omega))||_N \mathrm{d}s \bigg)\\
& \leq \frac{E T}{\widehat U_N N^{1/4}}\left(\sqrt{A_1(T)}+ \sqrt{A_2(T)} \right)
\end{align*}
\end{proof}
Following the lines in~\cite{bcelm} the next step is to estimate (\ref{flow_estimate}).
\begin{proposition}\label{proppathbdd}
\begin{equation}
|||D M_{t,s}(\widehat{\textbf{V}} (s,\textbf{V}_0,\omega))|||\leq e^{t\lambda(\widehat{\textbf{V}}(s,\textbf{V}_0,\omega))},
\end{equation}
where
\begin{equation}
\lambda(\textbf{V})=\frac{4}{\sqrt{U(\textbf{V})}}.
\end{equation}
\end{proposition}
\begin{proof}
Let $s<t_1<t$ and $\textbf{X}\in \R^N$. Consider the expression
$$\Psi_{t-t_1}\circ R_{ij}(\theta_1) \circ \Psi_{t_1-s}(\widehat{\textbf{V}}_s)$$
which is a part of (\ref{pathstochNp}). We have
\begin{align*}
&||D (\Psi_{t-t_1}\circ R_{ij}(\theta_1) \circ \Psi_{t_1-s})(\widehat{\textbf{V}}_s) \textbf{X}||
\leq |||D (\Psi_{t-t_1}\circ R_{ij}(\theta_1) \circ \Psi_{t_1-s})(\widehat{\textbf{V}}_s)||| \hspace{0.015in}||\textbf{X}|| \\
&\leq |||(D \Psi_{t-t_1})(R_{ij}(\theta_1) \circ \Psi_{t_1-s}(\widehat{\textbf{V}}_s)) |||
\hspace{0.015in} |||(D R_{ij}(\theta_1))(\Psi_{t_1-s}(\widehat{\textbf{V}}_s))||| \hspace{0.015in} |||(D \Psi_{t_1-s})(\widehat{\textbf{V}}_s)||| \hspace{0.015in}||\textbf{X}||.
\end{align*}
The fact that $\|\Psi_t(\textbf{V})\| =\|\textbf{V}\|$ and that $DR_{ij}(\theta_1)$ is norm preserving implies that
%
%
%
\begin{equation*}
|||(D R_{ij}(\theta_1))(\Psi_{t_1-s}(\widehat{\textbf{V}}_s))\textbf{X}|||\leq ||\textbf{X}||,
\end{equation*}
which in turn leads to
\begin{equation*}
|||D \Psi_{t-t_1} \circ R_{ij}(\theta_1) \circ \Psi_{t_1-s}||| \leq |||D \Psi_{t-t_1}||| \hspace{0.02in} |||D \Psi_{t_1-s}|||.
\end{equation*}
By repeating this procedure $n$ times we get
\begin{equation}\label{bddforM}
|||D M_{t,s}(\widehat{\textbf{V}}(s,\textbf{V}_0,\omega))|||\leq \prod_{k=1}^{n}|||D\Psi_{t_k-t_{k-1}}|||.
\end{equation}
To estimate the right hand side of the last inequality, we begin by noting that
\begin{equation*}
\frac{\mathrm{d}}{\mathrm{d}t}D\Psi_t(\textbf{V})=D\textbf{F}(\Psi_t(\textbf{V}))D\Psi_t(\textbf{V})
\end{equation*}
with $D\Psi_0(\textbf{V})=Id$. Next,
\begin{align*}
\frac{\mathrm{d}}{\mathrm{d}t}||D\Psi_t(\textbf{V})\textbf{X}||^2=
2 <\frac{\mathrm{d}}{\mathrm{d}t}D\Psi_t(\textbf{V})\textbf{X}\hspace{0.01in} ,\hspace{0.015in} D\Psi_t(\textbf{V})\textbf{X}>\\
=2<D\textbf{F}D\Psi_t(\textbf{V})\textbf{X}\hspace{0.01in} ,\hspace{0.01in} D\Psi_t(\textbf{V})\textbf{X}>.
\end{align*}
Differentiating the left hand side of the last inequality and  using the Cauchy Schwartz inequality yields
\begin{align}\label{bddfornormphi}
\frac{\mathrm{d}}{\mathrm{d}t}||D\Psi_t(\textbf{V})\textbf{X}||\leq ||D\textbf{F}||_*
\hspace{0.03in}||D\Psi_t(\textbf{V}) \textbf{X} ||
\end{align}
where
\begin{align*}
||D\textbf{F}||_*:= \sup_{||\textbf{X}||=1}|<D\textbf{F}\hspace{0.01in}\textbf{X}\hspace{0.01in},\hspace{0.01in} \textbf{X}>|.
\end{align*}
By (\ref{Nparticle_Ffield}), we have
\begin{align*}
\frac{\partial}{\partial v_i} \textbf{F}_j =
-\frac{E}{U(\textbf{V})}\frac{v_j}{N}+ E\frac{J(\textbf{V})}{U(\textbf{V})^2}\frac{2v_i v_j}{N}-E \frac{J(\textbf{V})}{U(\textbf{V})}\delta_{ij},
\end{align*}
where
\begin{equation*}
\delta_{ij} =
\begin{cases}
1 & \text{if } i=j,
\\
0 & \text{if } i   \neq j.
\end{cases}
\end{equation*}
Writing $\textbf{X}=(x_1,\dots,x_N)$ with $||\textbf{X}||=1$, we have
\begin{align*}
&||D\textbf{F}(\textbf{V})||_* \leq \frac{E}{U(\textbf{V}) N}
\left |\sum_{i=1}^{N}\big( \sum_{j=1}^{N}v_jx_j \big )x_i \right | +
\frac{2E|J(\textbf{V})|}{U(\textbf{V})^2 N}
\left |\sum_{i=1}^{N}\big( \sum_{j=1}^{N} v_iv_jx_j \big )x_i \right | \\
&\hspace{0.8in}+\frac{E|J(\textbf{V})|}{U(\textbf{V})}
\left |\sum_{i=1}^{N}\big( \sum_{j=1}^{N} \delta_{ij}x_j \big )x_i \right | :=A+B+C.
\end{align*}
Using the Cauchy-Schwartz inequality twice yields
\begin{align*}
\left |\sum_{i=1}^{N}\big( \sum_{j=1}^{N}v_jx_j \big )x_i \right | \leq \sqrt{N} \left |\sum_{j=1}^{N}v_jx_j \right | \leq
\sqrt{N}||\textbf{V}|| =N \sqrt{U(\textbf{V})}.
\end{align*}
Hence
\begin{equation*}
A\leq \frac{E}{\sqrt{U(\textbf{V})}}.
\end{equation*}
The inequality $|J(\textbf{V})|\leq \sqrt{U(\textbf{V})}$ together
with the definition of $U(\textbf{V})$ leads to
\begin{equation*}
B\leq \frac{2E |J(\textbf{V})|}{U(\textbf{V})^2 N}||\textbf{V}||^2\leq \frac{2E}{\sqrt{U(\textbf{V})}},
\end{equation*}
and
\begin{equation*}
C\leq \frac{E}{\sqrt{U(\textbf{V})}}.
\end{equation*}
Collecting all these inequalities, we finally obtain
\begin{equation*}
||D\textbf{F}(\textbf{V})||_* \leq \frac{4E}{\sqrt{U(\textbf{V})}}.
\end{equation*}
Plugging this into (\ref{bddfornormphi}), for all $\textbf{X} \in \R^N$, we have
\begin{equation*}
\frac{\mathrm{d}}{\mathrm{d}t}||D\Psi_t(\textbf{V})\textbf{X}||\leq \frac{4E}{\sqrt{U(\textbf{V})}} ||D\Psi_t(\textbf{V})\textbf{X}||.
\end{equation*}
Solving this differential inequality yields
\begin{equation*}
|||D\Psi_t(\textbf{V})||| \leq \exp \left( 4Et/\sqrt{U(\textbf{V})} \right ).
\end{equation*}
Applying the last inequality to (\ref{bddforM}), we conclude that
\begin{equation*}
||D M_{t,s}(\widehat{\textbf{V}}(s,\textbf{V}_0,\omega))||\leq
\exp \left(\sup_{0\leq s \leq t} 4Et / \sqrt{U(\widehat{\textbf{V}}(s,\textbf{V}_0,\omega))} \right ).
\end{equation*}
\end{proof}
In order to complete the estimate for (\ref{diffpath}), we actually also need to show that, for large $N$, the probability that $\sup_{0\leq s\leq t}U(\widehat{\textbf{V}})^{-1/2}$ is large is small. This is because, while the quantity $U({\textbf{V}})$ is conserved by the master equation
(\ref{Nparticlemaster}), the quantity $U(\widehat{\textbf{V}})$ is not conserved by the quenched master equation (\ref{quenchmaster}). The following lemma is based on~\cite{bcelm} with a small difference in that, here we only have that the quenched master equation equation propagates chaos and not independence.
\begin{lemma}\label{bddrooten}
Let $\widehat W_N(\textbf{V},0)$ be a probability density on $\R^N$ satisfying the assumptions in Theorem \ref{thm_quant_thermostat} and
$$ \widehat U_N >0 \hspace{0.3in}\mbox{and}\hspace{0.3in} \widehat m_{6,N}(0)<\infty .$$
Then for $0<t<T$, we have
\begin{equation}
\prob\left \{\sup_{0\leq s\leq t}U(\widehat{\textbf{V}})^{-1/2}\geq 2\sqrt{\frac{2}{U_N}} \right \}
\leq \frac{4}{\widehat U_N^2 \sqrt{N}}A_2(T)n(T),
\end{equation}
where $A_2(T)$ is given by Proposition \ref{propbddforce} and $n(t)$
is the smallest integer such that $n(t)\geq \frac{t}{\delta_t}+1$ with
$\delta_t $ defined by
\begin{equation}
\delta_t:=\frac{\sqrt{\widehat U_N}}{E} \log \left(
    \frac{2+2\sqrt{2}}{1+2\sqrt{2}} \right)
\end{equation}
\end{lemma}
\begin{proof}
This proof can be carried out almost as in~\cite{bcelm}, but with some modification to account for the lack of independence. For completeness we present the full proof, not only the needed modifications. From the definitions, we have
\begin{equation*}
\frac{d}{d t}U(\widehat{\textbf{V}}(t))=2EJ(\widehat{\textbf{V}}(t))-2E\frac{\widehat J_N(t)}{\widehat U_N}U(\widehat{\textbf{V}}(t)).
\end{equation*}
Using the inequalities $|J(\widehat{\textbf{V}}(t)|\leq \sqrt{U(\widehat{\textbf{V}}(t))}$ and
$|\widehat J_N(t)|\leq \sqrt{\widehat U_N}$, we obtain
\begin{align*}\label{bdddiffofU}
\nonumber \frac{d}{d t} \left(U(\widehat{\textbf{V}}(t))\right)^{-1/2}&=-\frac{1}{2}\left(U(\widehat{\textbf{V}}(t))\right)^{-3/2}
\frac{d}{d t}U(\widehat{\textbf{V}}(t))\\
&\leq E\left(U(\widehat{\textbf{V}}(t))\right)^{-1}+
\frac{E}{\sqrt{\widehat
    U_N}}\left(U(\widehat{\textbf{V}}(t))\right)^{-1/2}.
\end{align*}
Writing
\begin{equation*}
x(t)=\frac{1}{\sqrt{U(\widehat{\textbf{V}}(t))}},
\end{equation*}
the last differential inequality corresponds to the following differential equation
\begin{equation*}
x'(t)=\frac{E}{\sqrt{\widehat U_N}}x(t)+E x^2(t),
\end{equation*}
with initial condition $ x(t_0)=x_0$. The solution is given by
\begin{equation*}
x(t)=\left(\frac{1}{x_0}e^{-(E/\sqrt{\widehat
      U_N})(t-t_0)}-\sqrt{\widehat
      U_N}\left(1-e^{-(E/\sqrt{\widehat U_N})(t-t_0)}\right)
\right)^{-1}.
\end{equation*}
The above solution $x(t)$ blows up in finite time. However, we can
still hope that the solutions starting at $x_0$ do not blow up in a
time interval whose length is independent of $t_0$. To be more
precise, let $t_1$ denote the time at which $x(t_1)=2x_0$. Then
\begin{equation}
e^{-(E/\sqrt{\widehat U_N})(t_1-t_0)}=\frac{1+2\sqrt{\widehat U_N}x_0}{2+2\sqrt{\widehat U_N}x_0}.
\end{equation}
Choosing $x_0=\sqrt{2/\widehat U_N}$ leads to
\begin{equation}\label{delta_t}
t_1-t_0=\frac{\sqrt{\widehat U_N}}{E} \log \left(\frac{2+2\sqrt{2}}{1+2\sqrt{2}}\right).
\end{equation}
The length of the interval $[t_0,t_1]$ is independent of $t_0$ since
$E$ and $\widehat U_N$ are given. Thus, we now have that if
$\left(U(\widehat{\textbf{V}}(t_0))\right)^{-1/2}\leq \sqrt{2/\widehat
  U_N}$, then for all $t$ in $[t_0,t_1]$,
$\left(U(\widehat{\textbf{V}}(t))\right)^{-1/2}\leq 2\sqrt{2/\widehat
  U_N}$.\\
Moreover, for any given $t_0<T$ with $T$ from Theorem \ref{thm_quant_thermostat}, using Corollary \ref{corthmquench} we get
\begin{align*}
\E\left(U(\widehat{\textbf{V}}(t_0))-\widehat U_N \right)^2&\leq \frac{A_2(T)}{\sqrt{N}},
\end{align*}
where $A_2(T)$ is given by Proposition \ref{propbddforce}.
By the Chebychev inequality, we now have

\begin{align*}\label{bddU0}
\nonumber&\prob \left\{(U(\widehat{\textbf{V}}(t_0)))^{-1/2}>\sqrt{2/\widehat U_N} \right\}=
\prob \left\{U(\widehat{\textbf{V}}(t_0))<\frac{\widehat U_N}{2} \right\}\\
&\leq \prob \left\{|U(\widehat{\textbf{V}}(t_0))-\widehat U_N|>\frac{\widehat U_N}{2} \right\}
\leq \frac{4}{\widehat U_N^2 \sqrt{N}}A_2(T) .
\end{align*}
Hence, for large $N$, the probability that $U(\widehat{\textbf{V}}(t_0))^{-1/2}>\sqrt{2/\widehat U_N}$ is small.\\
Let now $\delta_t=t_1-t_0$ where $t_1-t_0$ is given by (\ref{delta_t}). Furthermore, for any given $t>0$, we set $n(t)$ to be the smallest integer such that $n(t)\geq \frac{t}{\delta_t}+1$. It now follows, if
$$ \left(U(\widehat{\textbf{V}}(j\delta_t))\right)^{-1/2}\leq \sqrt{2/\widehat U_N} \hspace{0.3in} \text{for all}\hspace{0.3in} 0\leq j\leq n(t), $$
we have by the reasoning above that
$$ \left(U(\widehat{\textbf{V}}(s)\right)^{-1/2}\leq 2\sqrt{2/\widehat U_N} \hspace{0.3in} \text{for all}\hspace{0.3in} 0\leq s\leq t. $$
Using this, we now get
\begin{align*}
&\prob\left \{\sup_{0\leq s \leq t} \left(U(\widehat{\textbf{V}}(s)\right)^{-1/2}\geq 2\sqrt{2/\widehat U_N}\right \}
\leq \prob\left \{\bigcup_{j=1}^{n(t)}\left \{ \left(U(\widehat{\textbf{V}}(j\delta t))\right)^{-1/2}\geq \sqrt{2/\widehat U_N}\right \} \right\}\\
&\leq \sum_{j=0}^{n(t)}\prob\left \{ \left(U(\widehat{\textbf{V}}(j\delta t))\right)^{-1/2}\geq \sqrt{2/\widehat U_N}\right \}.
\end{align*}
This implies for $t<T$ that
\begin{align*}
&\prob\left \{\sup_{0\leq s \leq t} \left(U(\widehat{\textbf{V}}(s)\right)^{-1/2}\geq 2\sqrt{\frac{2}{\widehat U_N}}\right \}
\leq \frac{4}{\widehat U_N^2 \sqrt{N}}A_2(T)n(T).
\end{align*}
This is what we wanted to show.
\end{proof}
Combining Proposition \ref{proppathbdd} with the last lemma leads to following corollary.
\begin{corollary}\label{bddprobenergy}
Let $\widehat W_N(\textbf{V},0)$ be a probability density on $\R^N$ satisfying the assumptions in Theorem \ref{thm_quant_thermostat} and
$$ \widehat U_N >0 \hspace{0.3in}\mbox{and}\hspace{0.3in} \widehat m_{6,N}(0)<\infty .$$
Let
\begin{equation}
\lambda(\widehat{\textbf{V}}(s,\textbf{V}_0,\omega)) =\frac{4E}{\sqrt{U(\widehat{\textbf{V}}(s,\textbf{V}_0,\omega))}}.
\end{equation}
Then for $0<t<T$, we have
\begin{align*}
&\prob\left \{\sup_{0\leq s \leq t} e^{t\lambda(\widehat{\textbf{V}}(s,\textbf{V}_0,\omega))} \geq e^{8t\sqrt{\frac{2}{\widehat U_N}}}\right \}
\leq\frac{4}{\widehat U_N^2 \sqrt{N}}A_2(T)n(T)
\end{align*}
with $A_2(T)$ given by Proposition $\ref{propbddforce}$ and $n(T)$ by Lemma \ref{bddrooten}.
\end{corollary}
\begin{proof}
Since the exponential function is increasing, the proof follows from the last Lemma.
\end{proof}
We are now ready to prove the main result of this section which again
is a modification of the corresponding result in~\cite{bcelm}:
\begin{theorem}\label{thmdistance}
Let $\widehat W_N(\textbf{V},0)$ be a probability density on $\R^N$ satisfying the assumptions in Theorem \ref{thm_quant_thermostat} and
$$ \widehat U_N >0 \hspace{0.3in}\mbox{and}\hspace{0.3in} \widehat m_{6,N}(0)<\infty .$$
Then for all $\epsilon >0$,
\begin{align*}
&\prob \left \{ ||\textbf{V}(t,\textbf{V}_0,\omega)- \widehat{\textbf{V}}(t,\textbf{V}_0,\omega)||_N >\epsilon \right \}\leq \\
&\hspace{0.3in}\frac{1}{\epsilon}e^{8t\sqrt{\frac{2}{\widehat U_N}}}\frac{E T}{\widehat U_N N^{1/4}} \left(\sqrt{A_1(T)}+ \sqrt{A_2(T)} \right)+
\frac{4}{\widehat U_N^2 \sqrt{N}}A_2(T)n(T),
\end{align*}
with $A_1(T)$, $A_2(T)$ are given by Proposition $\ref{propbddforce}$ and $n(T)$ by Lemma \ref{bddrooten}
\end{theorem}
\begin{proof}
Following the lines of~\cite{bcelm}, we define two events $A$ and $B$, where, $A$ is the event such that
\begin{equation*}
\sup_{0\leq s \leq t} e^{t\lambda(\widehat{\textbf{V}}(s,\textbf{V}_0,\omega))} > e^{8t\sqrt{\frac{2}{\widehat U_N}}},
\end{equation*}
and $B$ is the event such that
\begin{equation*}
||\textbf{V}(t,\textbf{V}_0,\omega)- \widehat{\textbf{V}}(t,\textbf{V}_0,\omega)||_N\ \geq \epsilon.
\end{equation*}
To estimate $\prob(B)$, note that
\begin{equation}
\prob (B) \leq \prob(A) + \prob (B \cap A^c).
\end{equation}
From Corollary \ref{bddprobenergy} we obtain
$$\prob(A) \leq\frac{4}{\widehat U_N^2 \sqrt{N}}A_2(T)n(T).$$
To estimate $\prob (B \cap A^c)$, we first note that on $A^c$, by (\ref{pathdiff}) and Proposition \ref{proppathbdd} it follows that
\begin{equation*}
||\textbf{V}(t,\textbf{V}_0,\omega)- \widehat{\textbf{V}}(t,\textbf{V}_0,\omega)||_N \leq
e^{8t\sqrt{\frac{2}{\widehat U_N}}} \int_0^t ||\textbf{F}(\widehat{\textbf{V}}(s,\textbf{V}_0,\omega))-\widehat{\textbf{F}}(\widehat{\textbf{V}}(s,\textbf{V}_0,\omega))||_N\mathrm{d}s.
\end{equation*}
Using the Markov inequality, we get
\begin{align*}
\prob (B \cap A^c) &\leq \prob \left \{e^{8t\sqrt{\frac{2}{\widehat U_N}}} \int_0^t ||\textbf{F}(\widehat{\textbf{V}}(s,\textbf{V}_0,\omega))-\widehat{\textbf{F}}(\widehat{\textbf{V}}(s,\textbf{V}_0,\omega))||_N\mathrm{d}s  \geq \epsilon\right \}\\
&\leq \frac{1}{\epsilon}e^{8t\sqrt{\frac{2}{\widehat U_N}}}
\E \left( \int_0^t ||\textbf{F}(\widehat{\textbf{V}}(s,\textbf{V}_0,\omega))-\widehat{\textbf{F}}(\widehat{\textbf{V}}(s,\textbf{V}_0,\omega))||_N \mathrm{d}s \right).
\end{align*}
By Proposition \ref{propbddforce} we get
\begin{align*}
\prob (B \cap A^c) &\leq
\frac{1}{\epsilon}e^{8t\sqrt{\frac{2}{\widehat U_N}}}\frac{E
  T}{\widehat U_N N^{1/4}} \left(\sqrt{A_1(T)}+ \sqrt{A_2(T)} \right).
\end{align*}
Collecting the inequalities above, we conclude
\begin{align*}
&\prob \left \{ ||\textbf{V}(t,\textbf{V}_0,\omega)- \widehat{\textbf{V}}(t,\textbf{V}_0,\omega)||_N >\epsilon \right \}\leq \\
&\hspace{0.3in}\frac{1}{\epsilon}e^{8t\sqrt{\frac{2}{\widehat U_N}}}\frac{E T}{\widehat U_N N^{1/4}} \left(\sqrt{A_1(T)}+ \sqrt{A_2(T)} \right)+
\frac{4}{\widehat U_N^2 \sqrt{N}}A_2(T)n(T)
\end{align*}
\end{proof}

\section{Propagation of chaos for the master equation (\ref{Nparticlemaster})}
\noindent We are finally ready to show the main result of this paper, namely, that the second marginal $f_2^N(v_1,v_2,t)$ of
$W_N(\textbf{V},t)$ satisfying the master equation
(\ref{Nparticlemaster}) converges as $N\rightarrow \infty$ to the
product of two one marginals $f(v_1,t)f(v_2,t)$ of $W_N(\textbf{V},t)$
where $f(v,t)$ solves (\ref{thermostat_boltzmann}). In~\cite{bcelm},
the idea is to introduce two empirical distributions corresponding to
the two stochastic processes $\textbf{V}(t)$,
$\widehat{\textbf{V}}(t)$ and make use of the propagation of
independence to apply the law of large numbers. In our case,
independence between particles is not propagated, but the quenched
master equation (\ref{quenchmaster}) propagates chaos which together
with Theorem \ref{thmdistance} gives that, for large $N$ with high
probability the distance between the paths of the two stochastic
processes can be made arbitrary small.
We start by introducing the two following empirical distributions: For
each fixed $N$ and $t>0$, let
\begin{equation}\label{empN}
\mu_{N,t}=\frac{1}{N}\sum_{j=1}^{N}\delta_{v_j(t,\textbf{V}_0,\omega)}
\end{equation}
and
\begin{equation}\label{empq}
\widehat{\mu}_{N,t}=\frac{1}{N}\sum_{j=1}^{N}\delta_{\widehat{v}_j(t,\textbf{V}_0,\omega)}.
\end{equation}
Since we have shown that the master equation (\ref{quenchmaster}) propagates chaos, it follows from~\cite[Proposition 2.2]{sznitman} ,  that
\begin{equation}\label{chaosempiric}
\lim_{N\rightarrow \infty}\widehat{\mu}_{N,t}=\widehat f(v,t)\mathrm{d}v
\end{equation}
where the convergence is in distribution and $\widehat f(v,t)$ is the solution to (\ref{thermostat_boltzmann}), see~\cite[Theorem 2.1]{yosief}.\\
Theorem~\ref{thmdistance} shows that the distance between the two empirical measures above goes to zero as
$N \rightarrow \infty$. To be more precise, we need to recall the Kantorovich-Rubinstein Theorem (KRT) concerning
the $1$- Wasserstein distance. Let
$$\p_1(\R^N)=\left \{\mu  : \int |v|\mathrm{d}\mu(v) <\infty \right\}.$$
\begin{theorem}\label{thmkantorovic}
For any $\mu, \eta \in \p_1(\R^N)$
\begin{equation}
\W_1(\mu,\eta)=\sup\left\{\int_{\R^N}\phi(v)\mathrm{d}\mu(v)-\int_{\R^N}\phi(v)\mathrm{d}\eta(v) \right\}
\end{equation}
where the supremum is taken over the set of all $1$-Lipschitz
continuous functions $\phi:\R^n\rightarrow \R$.
\end{theorem}
We now have
\begin{lemma}\label{lemw1dist}
Let the assumptions of Theorem~\ref{thm_quant_thermostat}
be satisfied, and assume that $\widehat U_N>0$ and that
$\widehat m_{6,N}(0)<\infty$. Moreover let
 $\phi$ be a 1-Lipschitz function. For the $\W_1$ defined as in Theorem \ref{thmkantorovic}, and all $\epsilon >0$ we have
\begin{equation}
\lim_{ N\rightarrow \infty}\prob \left\{\W_1(\mu_{N,t},\widehat{\mu}_{N,t})\geq \epsilon  \right\}=0.
\end{equation}
\end{lemma}
\begin{proof}
The Lipschitz condition together with the Cauchy Schwartz inequality yields
\begin{align*}
&\left | \int_{\R}\phi(u)\mathrm{d}\mu_{N,t}(u)- \int_{\R} \phi(w)\mathrm{d}\widehat{\mu}_{N,t}(w) \right | \\
&\hspace{1.3in} =\left |\frac{1}{N}\sum_{j=1}^{N}\phi (v_j(t))-\frac{1}{N}\sum_{j=1}^{N}\phi(\widehat{v}_j(t)) \right |\\
&\hspace{1.3in} \leq \frac{1}{N}\sum_{j=1}^{N} | v_j(t) - \widehat{v}_j(t) | \\
&\hspace{1.3in} \leq \sqrt{\frac{1}{N}\sum_{j=1}^{N}| v_j(t) - \widehat{v}_j(t) |^2}\\
&\hspace{1.3in} = ||\textbf{V}(t,\textbf{V}_0,\omega)-\widehat{\textbf{V}}(t,\textbf{V}_0,\omega) ||_N.
\end{align*}
Using Theorem \ref{thmdistance} we get
\begin{equation*}
\prob \left\{\W_1(\mu_{N,t},\widehat{\mu}_{N,t})\geq \epsilon  \right\} \leq
\prob
\left\{||\textbf{V}(t,\textbf{V}_0,\omega)-\widehat{\textbf{V}}(t,\textbf{V}_0,\omega)
  ||_N\geq \epsilon  \right\}\rightarrow 0,
\hspace{0.1in}\text{when }N\rightarrow \infty.
\end{equation*}
\end{proof}
Using Lemma \ref{lemw1dist} and (\ref{chaosempiric}) together with the
fact that the quenched master equation propagates chaos (Theorem
\ref{thm_quant_thermostat}), we are now ready to prove our main
result.
\begin{theorem}\label{thmpropchaosNmaster}
Let $\widehat W_N(\textbf{V},0)$ be a probability density on $\R^N$
satisfying the assumptions in Theorem \ref{thm_quant_thermostat} and
$$ \widehat U_N >0 \hspace{0.3in}\mbox{and}\hspace{0.3in} \widehat m_{6,N}(0)<\infty .$$
Let
$\textbf{V}(t,\textbf{V}_0,\omega)=(v_1(t,\textbf{V}_0,\omega),\dots,v_N(t,\textbf{V}_0,\omega))$
be the stochastic process corresponding to the master equation
(\ref{Nparticlemaster}) with initial condition given by
(\ref{intialdataassumption}). Then for all 1-Lipschitz function $\phi$ on $\R^2$ with $||\phi||_{\infty}<\infty$ and all $t>0$ we have
\begin{equation}
\lim_{N\rightarrow \infty}\E \left [\phi(v_1(t,\textbf{V}_0,\omega),v_2(t,\textbf{V}_0,\omega))\right ] =
\int_{\R^2} \phi(v_1,v_2)f(v_1,t)f(v_2,t)\mathrm{d}v_1\mathrm{d}v_2
\end{equation}
where the expectation is with respect to the collision history $\omega$ and initial velocities $\textbf{V}_0$.\\
\end{theorem}
\begin{proof}
Since the probability density $\widehat W$ is symmetric under permutation, we have
\begin{equation}
\E \left[ \phi(v_1(t),v_2(t))\right] = \frac{1}{N-1} \sum_{j=2}^{N} \E \left[ \phi(v_1(t),v_j(t))\right].
\end{equation}
Let $\Psi(u)$ be defined by
\begin{equation}
\Psi(u)=\int_{\R}\phi(u,w)\mathrm{d}\widehat{\mu}_{N,t}(w).
\end{equation}
From the properties of $\phi$ it follows that $\Psi$ is 1-Lipschitz on $\R$ and $||\Psi||_{\infty}\leq ||\phi||_{\infty}$.
For a any given $\varepsilon >0 $, we now have
\begin{align*}
&\big | \E[\phi(v_1(t),v_2(t))] -\E[\Psi(v_1(t))]\big | \leq \\
&\hspace{0.2in} \leq \E \left |\frac{1}{N-1}\sum_{j=2}^{N}\phi (v_1(t),v_j(t))-\frac{1}{N}\sum_{j=1}^{N}\phi(v_1(t),\widehat{v}_j(t)) \right |\\
&\hspace{0.2in} \leq \E \left |\frac{1}{N-1}\sum_{j=2}^{N}\left [\phi (v_1(t),v_j(t))-\phi(v_1(t),\widehat{v}_j(t))\right] \right |+ \frac{2||\phi||_{\infty}}{N} \\
&\hspace{0.2in} = \E \left |\frac{1}{N-1}\sum_{j=2}^{N}\left [\phi (v_1(t),v_j(t))-\phi(v_1(t),\widehat{v}_j(t))\right]
\left[\mathbbm{1}_{\{||\textbf{V}-\widehat{\textbf{V}}||_N \geq\epsilon\}}+
\mathbbm{1}_{\{||\textbf{V}-\widehat{\textbf{V}}||_N <\epsilon\}}\right] \right | \\
&\hspace{0.4in} + \frac{2||\phi||_{\infty}}{N}\\
&\hspace{0.2in}\leq 2||\phi||_{\infty}\prob \left\{||\textbf{V}(t,\textbf{V}_0,\omega)-\widehat{\textbf{V}}(t,\textbf{V}_0,\omega) ||_N\geq \epsilon \right \}
+2\epsilon + \frac{2||\phi||_{\infty}}{N}.
\end{align*}
To obtain the last inequality we have used the 1-Lipschitz condition on $\phi$ and the Cauchy Schwartz inequality. A similar argument also yields
\begin{equation}
\left | \E[\Psi(v_1(t))] -\int_{\R}\Psi(v)\mathrm{d}\widehat{\mu}_{N,t}(v)\right |\leq
 2||\phi||_{\infty} \prob \left\{||\textbf{V}(t,\textbf{V}_0,\omega)-\widehat{\textbf{V}}(t,\textbf{V}_0,\omega) ||_N\geq \epsilon \right \}
 + \frac{2\|\phi\|_{\infty}}{N} +2\epsilon.
\end{equation}
Consulting Theorem \ref{thmdistance} and choosing $\epsilon=N^{-1/8}$ finally gives
\begin{align*}
&\left |\E[\phi(v_1(t),v_2(t))]-\int_{\R^2}\phi(v,w)\mathrm{d}\widehat{\mu}_{N,t}(v)\mathrm{d}\widehat{\mu}_{N,t}(w)\right|\\
&\hspace{0.1in}\leq \left | \E[\phi(v_1(t),v_2(t))] -\E[\Psi(v_1(t))]\right |+
\left | \E[\Psi(v_1(t))] -\int_{\R}\Psi(v)\mathrm{d}\widehat{\mu}_{N,t}(v)\right |\\
&\hspace{0.1in}\rightarrow 0 \hspace{0.2in}\text{when } N\rightarrow \infty.
\end{align*}
Since by (\ref{chaosempiric}) it follows that
\begin{equation}
\lim_{N\rightarrow \infty}\int_{\R^2}\phi(v,w)\mathrm{d}\widehat{\mu}_{N,t}(v)\mathrm{d}\widehat{\mu}_{N,t}(w)=
\int_{\R^2}\phi(v,w)f(v,t)f(w,t)\mathrm{d}v\mathrm{d}w,
\end{equation}
we conclude that
\begin{equation*}
\lim_{N\rightarrow \infty}\E \left [\phi(v_1(t,\textbf{V}_0,\omega),v_2(t,\textbf{V}_0,\omega))\right ] =
\int_{\R^2} \phi(v_1,v_2)f(v_1,t)f(v_2,t)\mathrm{d}v_1\mathrm{d}v_2\,,
\end{equation*}
which is what we wanted to show.
\end{proof}
\section{Acknowledgements}
We would like to thank an anonymous referee for having read the previous version of this paper very carefully and for pointing
out several important issues.
E.C. would like to thank Chalmers  Institute of Technology during a visit in the Spring of
2012, and would like to acknowledge support from N.S.F. grant DMS-1201354.
D.M. and B.W. would like to thank Kleber Carrapatoso for the
discussions during the initial phase of this work. This work was
supported by grants from the Swedish Science Council and the Knut and
Alice Wallenberg foundation.

\newpage
\vspace{2in}
\centering
Eric Carlen\\
Department of Mathematics, Hill Center,\\
Rutgers University \\
$110$ Frelinghuysen Road Piscataway NJ $08854-8019$ USA\\
email: carlen@math.rutgers.edu\\
\vspace{0.7in}
Dawan Mustafa\\
Department of Mathematical Sciences,\\
Chalmers University of Technology,\\
and Department of Mathematical Sciences,\\
University of Gothenburg,\\
SE $41296$ Gothenburg\\
email: dawan@chalmers.se\\
\vspace{0.7in}
Bernt Wennberg\\
Department of Mathematical Sciences,\\
Chalmers University of Technology,\\
and Department of Mathematical Sciences,\\
University of Gothenburg,\\
SE $41296$ Gothenburg\\
email: wennberg@chalmers.se


%
%
%
\end{document}